\documentclass[twocolumn, prb, nofootinbib, superscriptaddress]{revtex4}
\usepackage[T1]{fontenc}
\usepackage{graphicx}
\usepackage{amsmath}
\usepackage{amsfonts}
\usepackage{amssymb}
\usepackage{tikz}
\usepackage{amsthm}
\usepackage{physics}
\usepackage{xcolor}
\usepackage{algorithm}
\usepackage{algpseudocode}
\usepackage{hyperref}
\usepackage{soul}

\newtheorem{theorem}{Theorem}[section]

\newtheorem{lemma}[theorem]{Lemma}
\theoremstyle{definition}
\newtheorem{definition}{Definition}[section]

\newcommand{\X}{X}
\newcommand{\Z}{Z}

\newcommand{\x}{\sigma^X}
\newcommand{\z}{\sigma^Z}

\newcommand{\C}{K}

\begin{document}

\title{Quantum computing with anyons is fault tolerant}

\author{Anasuya Lyons}

\affiliation{Department of Physics, Harvard University, Cambridge, Massachusetts 02138, USA}

\affiliation{IBM Quantum, T. J. Watson Research Center, Yorktown Heights, New York 10598, USA}

\author{Benjamin J. Brown}

\affiliation{IBM Quantum, T. J. Watson Research Center, Yorktown Heights, New York 10598, USA}
\affiliation{IBM Denmark, Sundkrogsgade 11, 2100 Copenhagen, Denmark}

\begin{abstract}
    In seminal work~\cite{Kitaev2003} Alexei Kitaev proposed topological quantum computing~\cite{Kitaev01, Kitaev2003, Freedman2002, Nayak2008}, whereby logic gates of a quantum computer are conducted by creating, braiding and fusing anyonic particles on a two-dimensional plane. Furthermore, he showed the proposal is inherently robust to local perturbations~\cite{Kitaev01, Kitaev2003, Bravyi2010, Bravyi2011} when anyons are created as quasiparticle excitations of a topologically ordered lattice model prepared at zero temperature. Over the decades following this proposal there have been considerable technological developments towards the construction of a fault-tolerant quantum computer. Rather than maintaining some target ground state at zero temperature, a modern approach is to actively correct the errors a target state experiences, where we use noisy quantum circuit elements to identify and subsequently correct for deviations from the ideal state.
    We present an error-correction scheme that enables us to carry out robust universal quantum computation by braiding anyons. We show that our scheme can be carried out on a suitably large device with an arbitrarily small failure rate assuming circuit elements are below some threshold level of local noise. The error-corrected scheme we have developed therefore enables us to carry out fault-tolerant topological quantum computation using modern quantum hardware that is now under development.
\end{abstract}

\maketitle

\noindent {\it Introduction.-} Topological quantum computing~\cite{Kitaev01,Kitaev2003, Freedman2002, Nayak2008} was among the first ideas put forward to process encoded quantum information using imperfect hardware.
With experimental demonstrations of topological codes in recent years, including various quantum memories based on topological codes~\cite{google2021, google2023, Bluvstein_2023, Iqbal2024, Iqbal_2025}, as well as state preparation and anyon braiding with more exotic topological phases~\cite{iqbal-nonAbelian2024, Minev_2025, lo2026universaltopological}, one can imagine that an experimental demonstration of fault-tolerant quantum computation by braiding is on the horizon. Here, we show how to carry out error correction throughout a quantum computation by a universal set of anyon braiding and fusion measurements in a topological phase ~\cite{SU24,mochon_smaller_groups_2004, kitaev_finite_group_2007, cui_universal_2015, lo_universal_2025}. Furthermore, we prove that its failure rate vanishes on a sufficiently large device. 

Topologically ordered models suitable for quantum computing via anyon braiding and fusion will produce anyons that are able to absorb other particles. This key feature enables us to protect quantum information, as information about the particle types that are absorbed by logical anyons cannot be accessed by local operations while the anyons are separated over long distances. As such they form logical degrees of freedom robust to local noise. However, a system capable of producing anyons that absorb other particle types presents new challenges not common to conventional `stabilizer' error correction~\cite{gottesman1997stabilizer}. 
This is precisely because errors that create anyons can hide the location of other particles, obfuscating the syndrome data of the true error configuration needed to diagnose and correct errors. If these anyons are left untreated it may be impossible to correct certain common error configurations due to the proliferation of many hidden excitations. It is therefore important to rapidly identify and annihilate erroneous anyons. On the other hand, the apparatus used to measure the locations of erroneous anyons may be unreliable. We therefore require a `fault-tolerant' scheme for error correction that is robust to false measurement reading due to circuit errors. A fault-tolerant scheme for error correction must proceed with some care when decisions on how to correct anyons to avoid introducing additional physical errors due to false readout information.

\begin{figure}
\includegraphics[width=1\columnwidth]{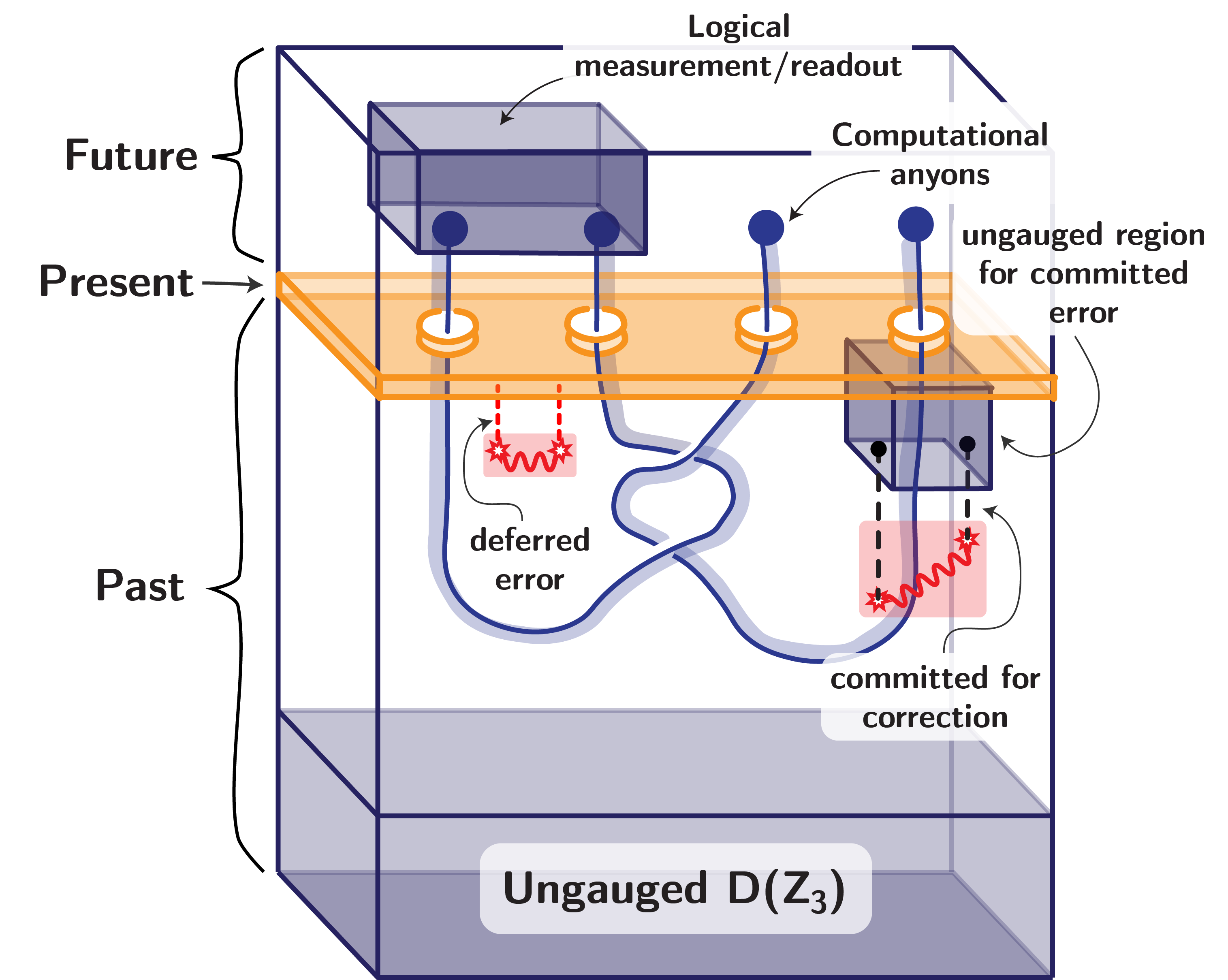}
\caption{Quantum computing with anyons displayed in spacetime where time runs upwards. The $D(S_3)$ phase is prepared by gauging an initial $D(\mathbb{Z}_3)$ state. Computational anyons (blue circles) are then prepared and braided while maintaining a very large separation. Errors, shown in red, acting on the topologically ordered lattice can be understood as string-like objects with erroneous anyons at their end points. Physical errors generate strings that lie perpendicular to the temporal direction, while measurement error strings run parallel to the time axis. When we are below the threshold error rate, errors form well-resolved clusters that, with high probability, are much smaller than the separation between computational anyons. We must correct errors in real time such that erroneous anyons cannot hide any important syndrome information. The just-in-time decoder takes syndrome information measured in the past and present timesteps to decide what correction to perform at each stage. When the decoder is confident a group of detection events should be mutually annihilated, the just-in-time-decoder commits to a correction. To make a correction, we ungauge the region supporting the anyons to reveal their collective fusion outcome and subsequently anihilate the erroneous anyons in the ungauged lattice before recovering the $D(S_3)$ phase. If the decoder is uncertain on how to correct errors, perhaps because they were detected recently, the correction can be deferred, as shown in red. Here, we correct the event with a measurement at the present time, moving the detection into the future. This allows us to delay correction. Ungauging can also be used to perform logical readout as shown at the top of the figure.
}
\label{fig:spacetime}
\end{figure}

We show how to conduct fault-tolerant error correction throughout a quantum computation by anyon braiding. We propose a just-in-time decoder~\cite{bombin20182dquantumcomputation3d, Brown_2020, davydova_universal_2025} running in concert with gauging operations~\cite{tanti2023hierarchy,verresen2022efficientlypreparingschrodingerscat, bravyi2022adaptiveconstantdepthcircuitsmanipulating,Lyons_2025} to perform reliable error-correction on a topological phase possessing a universal gate set through anyon braiding and fusion. The just-in-time decoder is an algorithm that assesses syndrome data in real time to make quick and confident decisions about how to correct erroneous anyons. At the same time, it defers some decisions until a later point to avoid catastrophic problems caused by false measurement readings. Once the decoder commits to performing a correction, we use gauging to correct errors. Gauging transforms a local region of our topologically ordered lattice, where we suspect an error to have occurred, into another type of topological code using local measurements and a constant-depth local unitary circuit. In this alternative phase, anyons can no longer absorb other anyons, and we can perform error correction operations using local rotations on individual degrees of freedom. Once we are satisfied that all errors are corrected have been corrected, we recover the original topological code used for computation with another gauging operation. All together, our scheme enables us to show that universal topological quantum computing can be accomplished with a vanishingly small error rate on a noisy quantum device (see Fig. \ref{fig:spacetime}). Earlier work towards the demonstration of anyon error correction has concentrated on numerical studies and analytical results for memories with non-universal anyon models~\cite{Wootton_2014,brell_thermalization_2014,wootton_active_2016, Dauphinais2017, schotte_fault-tolerant_2022,sala_decoherence_2024, davydova_universal_2025}. In Methods, we exemplify our result in microscopic detail using Kitaev's quantum double model $D(S_3)$, known to be universal for quantum computing by braiding and topological fusion measurements \cite{SU24,mochon_smaller_groups_2004, kitaev_finite_group_2007, cui_universal_2015, lo_universal_2025}.

\bigskip

\noindent {\it Topological quantum computing.-} We consider a two-dimensional lattice of spins that support anyonic excitations~\cite{Kitaev2003, Levin_2005}. Anyons~\cite{leinaas-myrheim, wilczek1982, goldin1985, Frohlich1990, wen1991, Nayak2008} are point-like quasiparticles created at the end points of operators with string-like support. We can also use string operators to move anyons. By moving two anyons close together, they combine in a fusion operation to produce a third anyon. In general, the product of a fusion operation is not deterministic. We write the fusion product of particle types labeled $a$ and $b$ as $a\times b = c + d + \dots$ where $c$, $d$ and so forth on the right-hand side of the equation denote the particle types that may appear as the outcome of the fusion operation.

Groups of anyons with multiple fusion outcomes give rise to a non-trivial Hilbert space that can be used for topological quantum computing. Given that anyons have to be moved close together to read their fusion outcome, by keeping anyons well separated the Hilbert space is inaccessible to local interactions and therefore remains protected. Additionally the protected Hilbert space can be manipulated by braiding anyons around one another. A topological quantum computation~\cite{Kitaev2003, Freedman2002} proceeds by first preparing computational anyons, separating them to encode information robustly, and then braiding them while maintaining their separation. The outcome of the computation is read out by fusion measurement (see Fig. \ref{fig:spacetime}). In general we can consider feedforward operations where a braiding or fusion operation in the future is chosen depending on the outcome of an earlier fusion measurement.

\bigskip

\noindent {\it Errors in general anyon models.-} Let us briefly discuss the challenges of correcting the errors of general anyon models before considering syndrome extraction using noisy circuit elements. We measure stabilizer operators to detect errors. Stabilizers project onto a configuration of anyonic excitations on the lattice whereby local errors are projected onto short string operators with anyonic excitations at their end points. Let us refer to the anyons created by errors as erroneous anyons, to distinguish them from computational anyons that are created deliberately to run a calculation. A configuration of erroneous anyons is commonly known as an error syndrome--- we illustrate example error syndromes that may occur in Fig.~\ref{Fig:LogicalError}. We will make frequent use of the following property to conduct error correction: given a region that was prepared to support a computational anyon, labeled $c$ and has experienced a local error that has produced a configuration of erroneous anyons $b_1$, $b_2$,\dots,$b_N$, we must have $ b_1 \times b_2 \times \dots \times b_N = c $ (fusion is associative).

\begin{figure}
\includegraphics[width=\columnwidth]{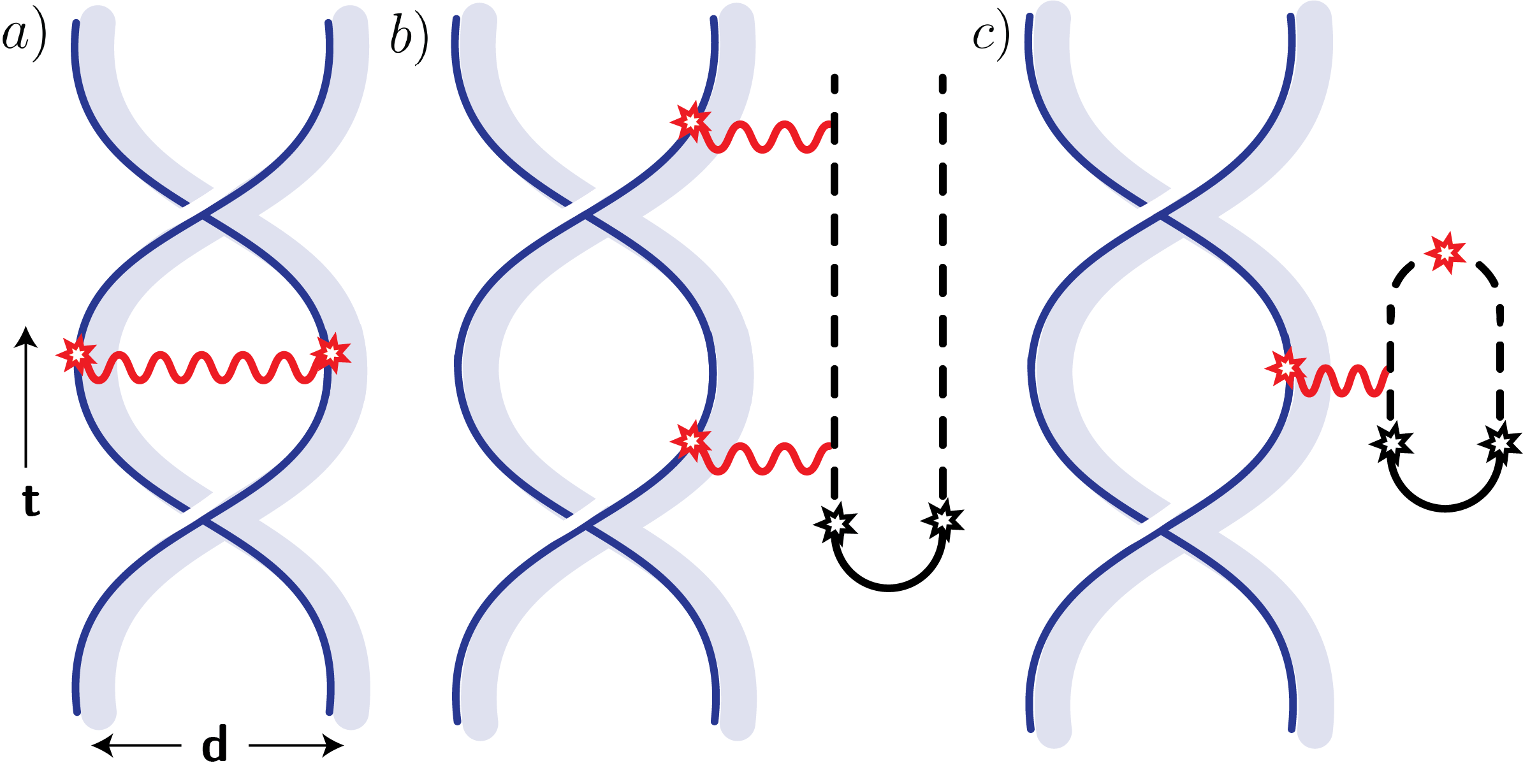}

\caption{\label{Fig:LogicalError} a) Part of a computation by braiding. Computational anyons are separated from each other by at least $d$ in all directions; erroneous anyons must propagate between computational anyons via a sequence of $O(d)$ local errors to introduce a logical error. b) A low-weight local error may cause a logical error over time. A small error first creates a pair of anyons, shown in black, nearby a computational anyon. A second local error (red) propagates an anyon between the computational anyon and an erroneous anyon, such that the erroneous anyon hides the syndrome of the second error. If the first pair of erroneous anyons are not annihilated, a subsequent third small local error can propagate an anyon from the erroneous anyon pair and a second computational anyon (moved during the course of braiding), introducing a logical error. We conclude that we must promptly annihilate erroneous anyons. c) If the erroneous pair is annihilated quickly, the charge exchanged with the computational anyon will be revealed, therefore enabling us to suppress the logical error.}
\end{figure}

\bigskip

\noindent {\it Detectors.-} We measure stabilizers using additional degrees of freedom and noisy circuit elements. We suppose that we can measure each stabilizer to detect any given species of anyon at a given position once per unit time.
In general, we expect circuit errors may give rise to false readings when we perform measurements. To deal with such errors, we measure stabilizers repeatedly over time and present syndrome information in a $2+1$-dimensional spacetime. By replacing stabilizer measurements $S_i$ with detectors $D_i(t)$, we can generalize the idealized $2d$ picture where erroneous anyons appear at the ends of error strings to this spacetime picture.

We make stabilizer measurements $S_i(t)$ of a generating set of operators $\{S_i\}$ over discrete time steps $t$ to measure the presence or absence of a certain excitation type at a particular location, indexed by $i$, of the two-dimensional lattice. Assuming that $S_i(t)$ may give a false outcome due to a measurement error, we instead use detectors $ D_i(t) := S_i(t) - S_i(t-1) $ for all generators $S_i$ and discrete timesteps $t$ to furnish our spacetime volume with syndrome data. Essentially, a detector compares the current reading of a stabilizer to its previous reading, and so represents changes in the stabilizer outcome.

Suppose a physical error introduced to the system at $t = T-1/2$ creates excitations detected by stabilizers $S_j$ and $S_k$. Then $D_j(T)$ and $D_k(T)$ will both register violations at time $T$. Assuming no additional errors occur, later readings of detectors $D_j(t > T)$ and $D_k(t > T)$ will not register further events. The error can then be interpreted as a string in spacetime extending along the spatial direction and can be corrected by fusing the two erroneous anyons measured by stabilizers $S_j$ and $S_k$.

Next, suppose measurement $S_i(t)$ gives a false reading: detectors $D_i(t)$ and $D_i(t+1)$ will both register detection events. This measurement error now looks like a string-like object that runs in the time direction, with detection events at its endpoints. Given that the state of the lattice has not changed, we can reverse this error simply by updating the classical readout of the stabilizer measurement. One can check that physical errors and measurement errors will compound to make longer string-like objects in the spacetime model. As a final remark on this topic we note that detectors are readily modified to account for the presence of a computational anyons. We simply adjust their value according to our knowledge of the locations of computational anyons.

\bigskip

\noindent {\it Decoding.-} Error correction with general anyon models introduces new challenges we do not see in more conventional error-correction schemes. Notably the presence of erroneous anyons can hide the locations of other excitations on the lattice. It is then important to annihilate anyons that may be hiding additional syndrome information. We therefore apply correction operations rapidly in real time, even before we have collected all of the syndrome data over the course of an entire computation. On the other hand, at the level of circuit noise, we can expect that anyons will be misidentified due to readout errors. We must be cautious not to fuse pairs of erroneous anyons too hastily, as measurement errors may cause us to mistakenly introduce new errors of our own making.

To remedy this issue we employ just-in-time decoding to correct errors that occur throughout a topological quantum computation. This is a decoder that takes as input the values of all detectors $D_i(t)$ at each timestep, and maintains the syndrome history of all detection events up until the present time. At each timestep the just-in-time decoder makes a decision on how to correct each detection event. It can make one of two actions for each detection event; it may either {\it commit} to a correction or {\it defer} its decision. When the decoder is confident in a group of detection events, it decides to commit and correct the corresponding excitations. Generically, if the decoder does not commit to a correction, a detection event is deferred. It is assumed that a measurement error occurred at the most recent timestep at the location where the excitation was detected. Reversing the most recent measurement effectively moves the detection event to a later timestep, and the decoder can commit this detection to a correction at a later timestep.

The just-in-time decoder is completely specified once we describe the criterion under which we are confident enough to commit to a correction. We formalize this definition in Methods but, in short, we commit to correcting a group of anyons once they have all existed for a time equal to the maximum separation between their initial spacetime detections. Given a local circuit noise model, this condition prevents local errors from combining in such a way that they become unmanageably large, while also ensuring the decoder corrects errors sufficiently quickly so anyons do not survive long enough to obscure a detrimental portion of the spacetime syndrome. 

\bigskip

\begin{figure}

\includegraphics[width=\columnwidth]{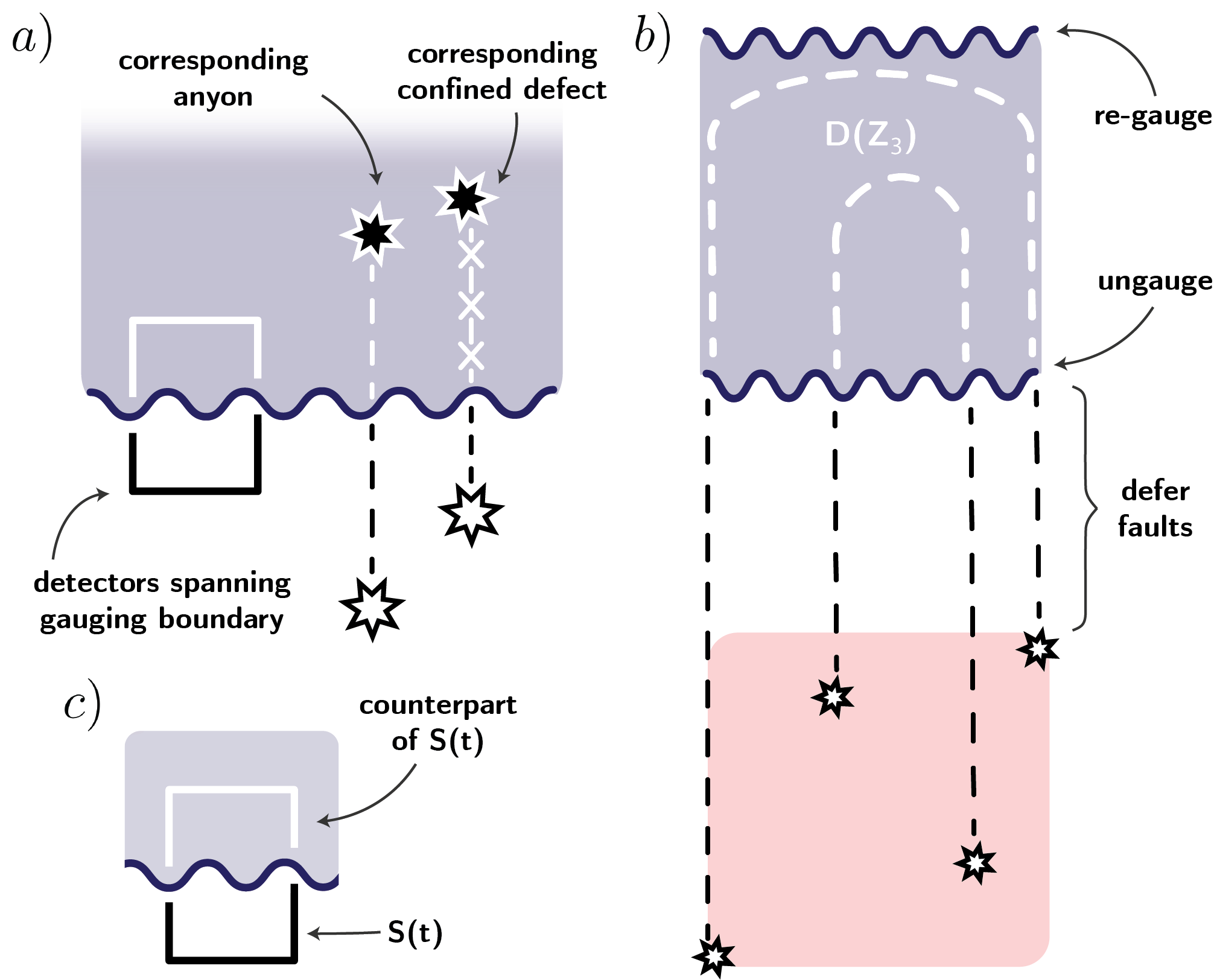}

\caption{a) Gauging and ungauging provide a dictionary between two codes, where each anyon has a counterpart. This allows us to construct meaningful detectors spanning the boundary. b) An error, contained in a local spacetime cluster (red) creates detection events (stars). The erroneous anyons are deferred for correction until a timestep when the decoder is confident to commit to a correction. To correct the error, we must use the properties of the anyon model. Local $D(S_3)$ stabilizer measurements alone cannot determine the fusion result of a configuration of erroneous anyons, as this is nonlocal information. Instead, the region is then ungauged, converting the nonlocal information into local data. After neutrality is evaluated, the $D(\mathbb{Z}_3)$ counterparts of the faults are fused. Finally, we gauge the lattice again to recover the $D(S_3)$ phase. c) These detectors are formed by taking pairs of stabilizers and their $D(\mathbb{Z}_3)$ counterparts in the ungauged phase.}
\label{Fig:gauging-domain-wall}
\end{figure}

\noindent {\it Error-correction by gauging.-} We now consider the specific actions the decoder must take upon committing to correcting a set of errors. We find that gauging and ungauging provide a practical scheme for performing these corrections. Gauging is a physical mapping between two distinct topologically ordered phases that respects certain consistency conditions between the anyonic excitations of the two phases. These consistency conditions enable excitations to pass through the domain wall between the two phases, though the particles may undergo some transformation (see Fig. \ref{Fig:gauging-domain-wall}a). Gauging has been considered previously for preparing topological phases \cite{verresen2022efficientlypreparingschrodingerscat, bravyi2022adaptiveconstantdepthcircuitsmanipulating, tanti2023hierarchy, lu_measurement_2022, li2023measuringtopologicalfieldtheories, Sukeno_2024} and for creating anyons \cite{bravyi2022adaptiveconstantdepthcircuitsmanipulating,Lyons_2025, ren_efficient_2025}. Here we use gauging to performing corrections. Important for error correction is the fact that the gauging map preserves the consistency of fusion outcomes between local collections of anyons: if a group of anyons fuses to $a$ in one phase, their gauged counterparts will fuse to the gauged counterpart of $a$ in the other phase. This enables us to read out the error syndrome in either phase.

We perform corrections by mapping the quantum double $D(S_3)$ onto the qutrit toric code $D(\mathbb{Z}_3)$. As $D(\mathbb{Z}_3)$ is a stabilizer code, its anyons have deterministic fusion outcomes. This means the fusion outcome of any group of anyons in the qutrit toric code can be determined from local stabilizer measurements. It also implies corrections take the form of low-depth circuit operations on local degrees of freedom. Once we commit to fusing a group of detection events in the spacetime syndrome history, we ungauge a small disk-shaped region containing all of the corresponding anyons that are committed to be fused together. This maps the $D(S_3)$ anyons in the region onto their counterparts in the qutrit toric code model.

In this ungauged region, we can now determine if the anyons on the region will fuse to a neutral configuration or, more generally, the species of computational anyon $c$ that was designated to appear on this region. The ungauging operation determines the anyon type on the region by revealing new information about the erroneous anyons. Essentially, the nonlocal information about the fusion result of the erroneous anyons has been converted into local information by the ungauging process, which can then be used to inform the decoder. If the fusion measurement reveals the expected computational anyon then we can proceed to prepare the $D(S_3)$ phase again having neutralized the error (Fig. \ref{Fig:gauging-domain-wall}b). Otherwise, we leave the region ungauged and expand the search, looking for other nearby erroneous anyons to merge into the ungauged region. The consistency conditions between the two phases of the ungauging map mean that we can produce detectors that straddle the domain wall between the two phases, such that we can consistently correct errors throughout the mapping (Fig. \ref{Fig:gauging-domain-wall}c). We ensure that the correction is robust to measurement errors by remaining in the ungauged phase for an extended period of time. After the decoder has cleaned up a given patch, we can recover $D(S_3)$ in the region by re-gauging.

\bigskip

\noindent {\it Statement of result.-} With our system and strategy for real-time error correction described, we can now state our result (proven in Methods). Assuming a local circuit noise model with sufficiently low rate of noise, we demonstrate that a universal quantum computation can be completed with a vanishingly small error rate by braiding and fusing anyons in the quantum double model. The local error model allows us to assume a cluster decomposition where local groups of errors are well resolved from all other errors with respect to their own length scale. We argue that our adaptive error correction strategy, using just-in-time decoding and ungauging to test neutrality, obtains corrections for local errors sufficiently quickly that an error cluster is not spread excessively far from its initial location in spacetime by further errors. Using the fact that very large error clusters are suppressed given a sufficiently low error rate, we find that the initial error plus the correction will not affect a quantum computation assuming the computational anyons involved in the calculation remain sufficiently well separated.  

\bigskip

\noindent {\it Discussion.-} We have shown that universal quantum computing with anyons is fault tolerant. Our results required the development of an error-correction scheme for the quantum double model $D(S_3)$ including a just-in-time decoder and a method for active correction using gauging operations. The development of our scheme is particularly timely given the recent demonstration of logic gates using this phase on a quantum device~\cite{lo2026universaltopological}. Our scheme can be utilized with these logic gates to show that their logical failure rate can be suppressed arbitrarily.

The field of quantum error correction seeks practical schemes for fault-tolerant quantum computing that operate with a relatively low overhead. Our error-correction scheme gives us a framework to design and implement numerical simulations to measure the resource cost of quantum computing with anyons. Such results will enable us to compare the resource cost of quantum computing with anyons to other two-dimensional fault-tolerant quantum computing schemes such as surface-code quantum computation supplemented by magic-state distillation~\cite{dennis_topological_2002, Bravyi05universal, fowler_surface_2012}, as well as more recent proposals for topological quantum computing that do not require braiding~\cite{bombin20182dquantumcomputation3d, Brown_2020, huang2025generatinglogicalmagicstates, davydova_universal_2025, sajith2025noncliffordgatesstabilizercodes}.

In this regard it might also be interesting to generalize our decoding scheme for other universal anyon models that are practical for realization with physical hardware. We expect that our method will generalize readily to other `solvable' anyon models that can be prepared via gauging \cite{tanti2023hierarchy, ren_efficient_2025}, as this will mean that our error-correction scheme using ungauging can be exploited. Beyond solvable anyon models, one might also consider generalizing our result to other anyon models where correction via gauging is not possible. We expect a correction can be obtained with the development of a new error-correction method where erroneous anyons are fused directly using ribbon operators or perhaps using anyonic interferometry~\cite{BONDERSON2008interferometry}. We leave it to future work to determine if other such schemes can be made fault tolerant.

Looking beyond two dimensions, there is currently a significant movement to design fault-tolerant quantum-computing architectures with low-density parity-check codes to reduce the resource overhead of quantum computing~\cite{Gottesman2014, Breuckmann2021quantum, Cohen2022, Bravyi2024high, Yoder2025tour}. A major challenge with such architectures is the realization of addressable non-Clifford gates acting on encoded quantum information. A recent proposal to deal with this issue suggests using code deformations, c.f. gauging, onto generalized `Clifford' stabilizer codes to prepare magic states~\cite{davydova_universal_2025}---see recent work~\cite{zhu2026non}. Indeed, in the search of quantum computing architectures with a high encoding rate and a novel logical gate set, one may even consider non-local generalizations of the quantum double of $D(S_3)$ by gauging a qutrit low-density parity check code. The methods we have developed here for codes in two dimensions can serve as a basis for the design of error-correction schemes for non-stabilizer generalizations of quantum low-density parity-check codes.

\bigskip

\noindent {\it Acknowledgments.-} We are grateful for helpful conversations with A. Bauer, M. Davydova, J. Magdalena de la Fuente, E. Lake, Y. Ren, and D.J. Williamson. BJB is grateful for the hospitality of the Center for Quantum Devices at the University of Copenhagen. AL acknowledges the NSF GRFP (Grant No. DGE 2140743) and the Kavli Institute for Theoretical Physics Graduate Fellowship Program, supported in part by Grant No. NSF PHY-1748958 to the Kavli Institute for Theoretical Physics (KITP), the Heising-Simons Foundation, and the Simons Foundation (216179, LB).

\bigskip

\noindent{\it Author Contributions.-} AL and BJB devised the research programme that led to the present results. AL produced the proof for the main theorem of this work. AL and BJB prepared the manuscript together.

\bigskip

\noindent{\it Competing Interests.-} No competing interests to declare.

\bigskip

\noindent{\it Materials and Correspondence.-} Correspondence and requests for materials should be addressed to \href{mailto:anasuya\_lyons@g.harvard.edu}{anasuya\_lyons@g.harvard.edu}.

\bibstyle{plain}
\bibliography{references}

\clearpage

\section*{Methods}

\renewcommand{\thetheorem}{\arabic{theorem}}
\renewcommand{\thedefinition}{\arabic{definition}}

\subsection{Quantum Double of $S_3$}

\begin{figure}
    \centering
    \includegraphics[width=0.85\linewidth]{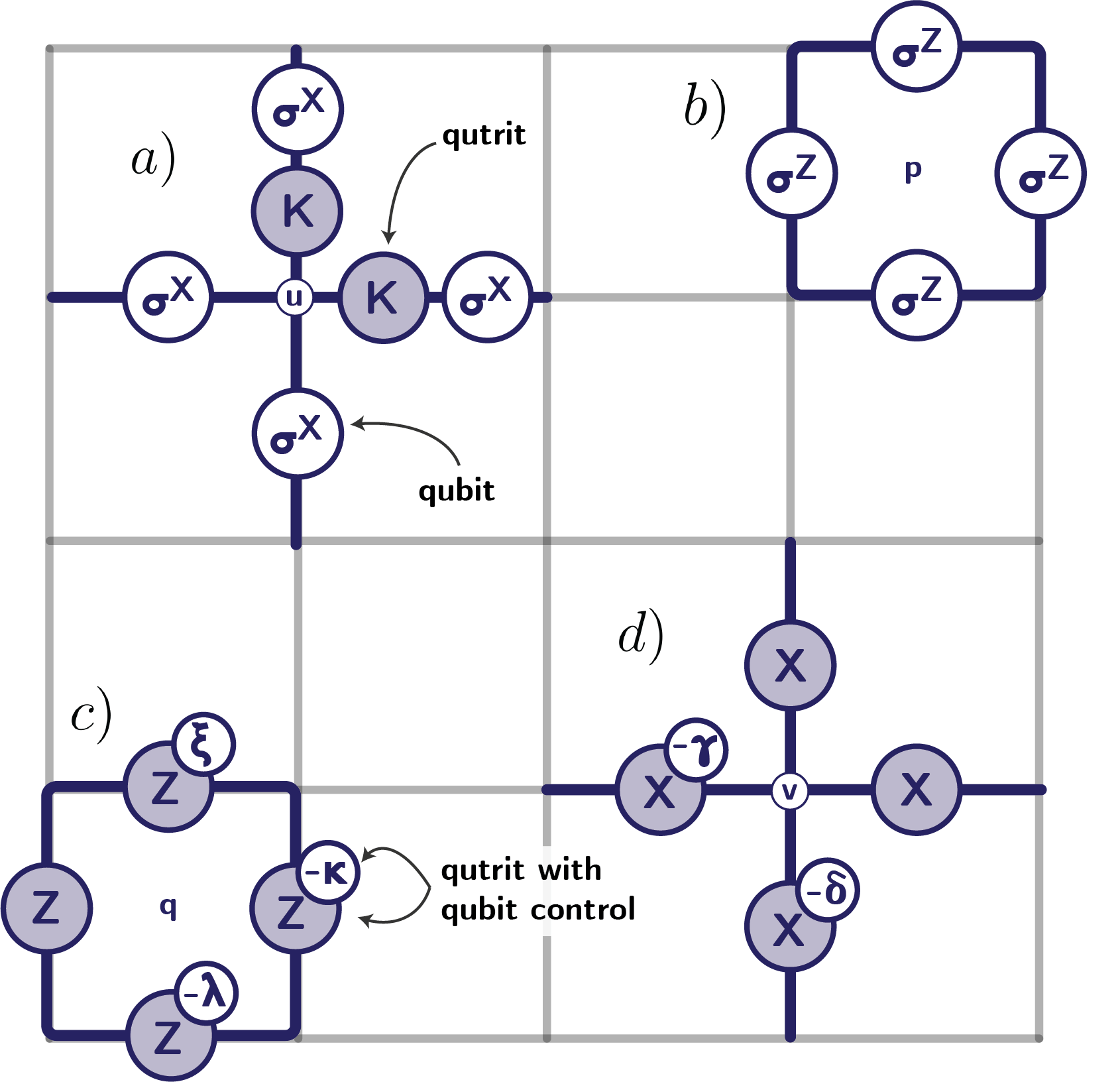}
    \caption{$D(S_3)$ lives on the lattice, with a local Hilbert space consisting of a qubit (unfilled circles) and a qutrit (shaded circles). A shaded circle with an unfilled circle above it indicates a qubit-control qutrit-target gate. a) Qubit vertex stabilizer $\alpha_u$ and b) qubit plaquette stabilizer $\beta_p$. c) Qutrit plaquette stabilizer $B_q$ and d) qutrit vertex stabilizer $A_v$. }
    \label{fig:stabilizers}
\end{figure}

We define the quantum double $D(S_3)$ based on the microscopic description given in Appendix~D of Ref.~\onlinecite{verresen2022efficientlypreparingschrodingerscat}, see also~\onlinecite{Brennen_2009, bravyi2022adaptiveconstantdepthcircuitsmanipulating}. Each edge supports a degree of freedom comprised of a three-level system and a two-level system, spanned by the states $| \ell , m \rangle $ with $\ell = 0,1,2$ and $m = 0,1$ (i.e. the tensor product of a qutrit and a qubit). We will follow a convention where capitalized Roman letters denote operators acting on the qutrit subsystem and Greek letters denote operators acting on the qubit subsystem. In general, operators acting on the qubit subsystem will commute with operators acting on the qutrit subsystem. Let us write generalized Pauli operators acting on the qutrit and qubit subsystems $\X$, $\Z$, and $\x$, $\z $, respectively, such that
\begin{equation}
\X | \ell , m \rangle = | (\ell + 1) \textrm{mod } 3 , m \rangle, \quad \Z | \ell ,  m \rangle = \omega^\ell | \ell  ,  m \rangle,
\end{equation}
where $\omega = \exp(2\pi i / 3)$ is the third root of unity and
\begin{equation}
\x | \ell , m \rangle = | \ell , (m+1) \textrm{mod } 2  \rangle ,\quad \z | \ell ,   m \rangle = (-1)^{m} | \ell  ,  m \rangle.
\end{equation}
We also have that $\Z \X = \omega \X \Z$, and $\z \x = - \x \z$. It is also helpful to note that $\X^3 = \Z^3 = ({\x})^2 = ({\z})^2 = 1$, and that $P^\dagger = P^{-1} = P^2$ for $P = \X,\Z$.

We also define the charge conjugation operator $\C$ supported on the qutrit subsystem. This hermitian operator satisfies the following relations under conjugation:
\begin{equation}
\C \X \C^\dagger = \X^\dagger,\quad \C \Z \C^\dagger = \Z^\dagger.
\end{equation}
The conjugation operator $\C$ commutes with operators acting on the qubit subsystem.

With the definitions above we can now write down the stabilizers of the quantum double $D(S_3)$. It can be helpful to view these stabilizers in terms of vertex and plaquette operators of disjoint codes on the qubit and qutrit subsystem, with decoration terms that couple the qubit and qutrit subsystem. We show the decorated vertex and plaquette terms for the qubit subsystem in Fig.~\ref{fig:stabilizers}(a) and~(b), denoted $\alpha_u$ and $\beta_p$, respectively. We have
\begin{equation}
\alpha_u =\C_{(u,N)} \x_{(u,N)} \C_{(u,E)}  \x_{(u,E)} \x_{(u,S)} \x_{(u,W)}, 
\end{equation}
where edges are indexed $e=(u, R)$ where $R = N, E,S,W$ denotes the edge $e$ incident to $u$ pointing in the north $ N$, east $ E$, south $S$ or west $W$ direction relative to $u$. Likewise we have plaquette operators supported on the qubit subsystem only:
\begin{equation}
\beta_p =  \z_{(p,N)}   \z_{(p,E)} \z_{(p,S)} \z_{(p,W)},
\end{equation}
where we index edges $e = ( p, R)$ on the north $N$, east $E$, south $S$ or west $W$ boundary edge of plaquette $p$.

We also have decorated vertex and plaquette operators supported mainly on the qutrit subsystem--- see Fig. \ref{fig:stabilizers}c and d:
\begin{equation}
S_v = (A_v + A_v^\dagger)/2, \quad S_p = (B_p + B_p^\dagger)/2, 
\end{equation}
where
\begin{equation}
A_v = \X_{(v,N)} \X_{(v,E)} \X_{(v,S)} ^{-\delta(v)} \X_{(v,W)}^{-\gamma(v)}, \label{Eqn:Astab}
\end{equation}
and
\begin{equation}
B_p = \Z_{(p,N)}^{ \xi(p)} \Z_{(p,E)}^ {-\kappa(p)} \Z_{(p,S)} ^{-\lambda(p)} \Z_{(p,W)}, \label{Eqn:Bstab}
\end{equation}
with exponents $\gamma(v)$, $\delta(v)$, $\xi(p)$, $\kappa(p)$ and $\lambda(p)$ are diagonal Pauli terms acting on the qubit subsystem. Explicitly we have
\begin{equation}
\gamma(v) = \z_{(v,W)}, \quad \delta(v) = \z_{(v,S)}, \quad \xi(p) = \z_{(p,W)},
\end{equation}
\begin{equation}
\kappa(p) = \z_{(p,W)} \z_{(p,N)} \z_{(p,E)} \equiv \z_{(p,S)}\beta_p,
\end{equation}
and
\begin{equation} 
\lambda(p) =  \z_{(p,N)} \z_{(p,E)} \z_{(p,S)} \z_{(p,W)}  \equiv \beta_p. 
\end{equation}

\begin{table}
    \centering
    \begin{tabular}{|c|c|c|}
    \hline
         \textbf{Anyon} & \textbf{Dimension} & \textbf{Alternate Label} \\
         \hline
         $1$ & 1 &  A, $[+]$\\
         $\eta$ & 1 &  B, $[-]$\\
         $\mu$ & 3 &  D, $[C_2]$\\
         $\phi$ & 3& E, $[C_2, -]$\\
         $e$ & 2 &  C, $[2]$\\
         $m$ & 2 &  F, $[C_3]$\\
         $f$ & 2& G, $[C_3, \omega]$\\
         $g$ & 2& H, $[C_3, \bar{\omega}]$\\
         \hline
    \end{tabular}
    \caption{Correspondence between our labeling convention for $S_3$ anyons and others used in the literature.}
    \label{tab:anyon-labels}
\end{table}

All of these stabilizers are Hermitian with eigenvalues $\pm 1$. We can consider a Hamiltonian
\begin{equation}
H = -\sum_v S_v -\sum_p S_p - \sum_v \alpha_v   - \sum_p \beta_p,
\end{equation}
where the ground states are $+1$ eigenvalue eigenstates of all the stabilizer operators; $\alpha_v$, $\beta_p$, $S_v$ and $S_p$. We can also consider excited states of this Hamiltonian where individual Hamiltonian terms are violated.

We say an excitation appears at a vertex or plaquette if its respective stabilizer is violated. The type of stabilizer that is violated determines the species of the particle.
We have electric particles $e$, and $\eta$ associated to violated vertex stabilizers  $S_v$ or $\alpha_v$, respectively, and magnetic particles $m$ and $\mu$ associated to violated plaquette stabilizers $S_p$ or $\beta_p$ operator, respectively.
Note that we have chosen labels $\eta$ and $\mu$ for particles associated to the star and plaquette operators on the qubit subsystem as they stand for Greek words $\underline{\eta} \lambda \epsilon \kappa \tau \rho \iota \kappa \textrm{\it \'{o}}  \varsigma$ and $\underline{\mu} \alpha \gamma \nu \eta \tau \iota  \kappa \textrm{\it \'{o}}  \varsigma$ meaning electric and magnetic, respectively.

The fusion rules of these anyons are encoded by the commutation relations of the stabilizers we have already defined. Let us define the group commutator
\begin{equation}
[V,W] \doteq V W V^{-1} W^{-1}.
\end{equation}

First of all, it is easily checked that $\beta_p$ commutes with all stabilizers:
\begin{equation}
[\alpha_v, \beta_p] = [A_v, \beta_p] = [B_q, \beta_p] = [\beta_q, \beta_p] = 1,
\end{equation}
for any choice of $v,q,p$. We find some non-trivial commutators among other pairs of stabilizers.
\begin{equation}
[A_v,B_p] = \omega^{(1 - \beta_p)/2}, \label{Eqn:ABcommutator1}
\end{equation}
for the case where $v$ lies on the south-west corner of plaquette $p$. Otherwise we have that 
\begin{equation}
[A_v,B_p] = 1, \label{Eqn:ABcommutator2}
\end{equation}
for any other choice of $v$ and $p$.  Similarly, we have
\begin{equation}
[A_u,A_v] = [B_p, B_q]= 1, \label{Eqn:AABBcommutators}
\end{equation}
for any choice of $u,v,p,q$. Lastly, we find that
\begin{equation}
[\alpha_v, A_v] = A_v, \quad [\alpha_v, B_p] = B_p, \label{Eqn:alphaAalphaBcomm}
\end{equation}
where the latter holds for the case where $v$ lies on the south-west corner of $p$. Otherwise we have that $[\alpha_v, B_p] = 1$, and likewise $[\alpha_u, A_v] = 1$ for all $u \not=v$.

Eq. \ref{Eqn:ABcommutator1} tells us that the presence of a violated $\beta_p$ stabilizer prevents neighboring $A_v$ and $B_p$ terms from commuting.  In terms of anyonic excitations, this means we cannot measure the presence of $e$ and $m$ excitations ($S_v$, $S_p$ stabilizer violations) in the presence of a $\mu$ anyon (a $\beta_p$ stabilizer violation). This physics is expressed by the following fusion rule:
\begin{equation}
    \mu \times a = \mu, \quad a = e, m.
\end{equation}
We can identify the $e$ and $m$ particles hidden by $\mu$ anyons by correcting violated $\beta_p$ stabilizers in pairs. This correction corresponds to a fusion operation where a pair of $\mu$ anyons are combined. We have
\begin{equation}
    \mu \times \mu = 1 + e + m + e \times m.
\end{equation}
Here we have written $e \times m$ to denote all anyons that are the result of fusing an $e$ and an $m$ together (we will discuss these anyons below). Similarly, Eq. \ref{Eqn:alphaAalphaBcomm} tells us that the commutation relations between $A_v$ ($B_p$) and $\alpha_v$ are determined by $A_v$ ($B_p$) itself, meaning that $e$ ($m$) anyons can absorb both $\eta$ anyons and other $e$ ($m$) anyons.
\begin{equation}
    a \times a = 1 + \eta + a, \quad a \times \eta = a,  \quad a = e, m.
    \label{Eqn:eXeAndmXm}
\end{equation}
Note that $\eta$ anyons always fuse deterministically:
\begin{equation}
    \eta \times \eta = 1
\end{equation}
A point of terminology: anyons that have multiple fusion outcomes are `non-Abelian', while anyons with deterministic fusion are `Abelian'.

In addition to $e, m, \eta$, and $\mu$, we also have fermionic and parafermionic excitations, $\phi$, $f$ and $g$, where we note again that we have chosen the symbol $\phi$ to stand for the Greek word for fermion; $ \underline{\phi} \epsilon \rho \mu \iota \textrm{\it \'{o}} \nu \iota \textrm{\it o}$. We have that the fusion product of an electric and magnetic particle on the qubit subsystem gives rise to a fermionic excitation
\begin{equation}
\eta \times \mu = \phi,
\end{equation}
and we also have that
\begin{equation}
    \mu \times m = \mu + \phi, \qquad \mu \times e = \mu + \phi
\end{equation}

We also obtain parafermionic excitations $f$ and $g$ as the fusion product of $e$ and $m$ excitations. There are two distinct outcomes of this fusion process; $f$ and $g$. 
\begin{equation}
e \times m = f + g. \label{Eqn:eXm}
\end{equation}
These anyons fuse with themselves as follows:
\begin{equation}
f \times f = 1 + \eta + f, \quad g \times g = 1 + \eta + g.
\end{equation}
where we note a parallel with the fusion rules shown in Eq.~\ref{Eqn:eXeAndmXm}.
Conversely to Eq.~\ref{Eqn:eXm}, we have
\begin{equation}
f \times g = e + m.
\end{equation}
In general we have 
\begin{equation}
a \times b = c + d,
\end{equation}
where $a$, $b$, $c$, $d$ are all different charge types of the four species $e$, $m$, $f$ and $g$. We can distinguish parafermions $f$ and $g$ on a site consisting of a nearby vertex-plaquette pair $(v,p)$ with the following projectors
\begin{equation}
F_{v,p} = (1 - A^\dagger_vB_p - A_v B_p^\dagger) /2, 
\end{equation}
and
\begin{equation}
G_{v,p} = (1 - A_vB_p - A^\dagger_v B_p^\dagger) /2
\end{equation}
where $F_{v,p}$ ($G_{v,p}$) projects onto a site occupied by an $f$ ($g$) particle.

For completeness, we also have the following fusion rules
\begin{equation}
\phi \times \phi = 1 + e + m + e\times m,
\end{equation}
and
\begin{equation}
\mu \times \phi = \eta + e + m + e\times m.
\end{equation}
It may be worth comparing our labeling convention to other works; see Table~\ref{tab:anyon-labels}.

We claim that it is sufficient to concentrate on the electric and magnetic excitations of the model, and neglect all fermionic and parafermionic excitations. Indeed, we see from the fusion rules that the fermion $\phi$ can be expressed uniquely as the fusion product of $\mu$ and $\eta$. We can therefore neglect this particle by regarding it in terms of its composite parts. Likewise, parafermions $f$ and $g$ are the local product of $e$ and $m$ anyons. We can therefore express parafermions in terms of $e$ and $m$ excitations. In any case, the $B_p$, $S_v$ and $S_p$ operators will also detect $\phi, f, g$ anyons, and when a cluster containing $\phi, f, g$ anyons is ungauged, the (para)fermions in the cluster will be split into their constituent $\mu, e, m$ anyons.

\subsection{Gauging}

We make use of gauging throughout a fault-tolerant quantum computation using $D(S_3)$. It is used to prepare the topological phase $D(S_3)$ from the phase $D(\mathbb{Z}_3)$. In a similar vein, we can also create computational anyons in the $D(S_3)$ phase.
We also use it to perform fusion measurements. This enables us to read out the encoded space of our computational anyons, as well as to read out the total charge of a local configuration of erroneous anyons to determine their neutrality.

In condensed matter and high-energy physics, `gauging' is a process that translates between a model with a global symmetry and one with a `gauge freedom', which is a localized version of the global symmetry \cite{RevModPhys.51.659}. This localization works in the following way. By definition, applying a symmetry to a symmetric model doesn't change its properties. Thus, it is convention how we label different symmetry sectors: the only physically meaningful variable is the difference between conventions. By introducing new degrees of freedom to keep track of different conventions in different places, we are free to change our convention locally, as long as we also update our bookkeeping degrees of freedom. This is a `gauge transformation'. In essence, by adding in new degrees of freedom, we have transformed our global symmetry into a local transformation. If we then focus purely on the properties of these new gauge degrees of freedom, we can construct a `gauge' theory, which is the `gauged' version of our original model.

We can think of this gauging procedure as a code deformation of a quantum error-correcting code. We begin with a model with a transversal logic gate. This transversal gate is a global symmetry in our condensed-matter parlance. With the use of some additional physical qubits, which we refer to as a gauge field, gauging decomposes the global transversal gate into local terms. These new terms are included in the stabilizer group of a new code that results from our code deformation. Indeed, one could regard this code deformation as a weight-reduction~\cite{hastings2023quantumweightreduction, Cohen2022} for the transversal logical operator of interest. The standard procedure for gauging ensures the resulting code has a stabilizer group of projectors that commute on the code space of the gauged code.

\subsubsection{Ungauging $S_3$}

The quantum double of $S_3$ can be thought of as the result of gauging a symmetry of the $\mathbb{Z}_3$ (qutrit) toric code. Recall that the qutrit toric code can be described by the stabilizers generated by $S^{\mathbb{Z}_3}_v = A^{\mathbb{Z}_3}_v + A^{\mathbb{Z}_3 \dagger}_v$, $S^{\mathbb{Z}_3}_p = B^{\mathbb{Z}_3}_p + B^{\mathbb{Z}_3 \dagger}_p$, where $A^{\mathbb{Z}_3}_v$ and $B^{\mathbb{Z}_3}_p$ denote the following qutrit operators:
\begin{equation}
\begin{aligned}
    A^{\mathbb{Z}_3}_v &= \X_{(v, N)} \X_{(v, E)} \X^\dagger_{(v, S)} \X^\dagger_{(v, W)} \\
    B^{\mathbb{Z}_3}_p &= \Z^\dagger_{(p, N)} \Z_{(p, E)} \Z_{(p, S)} \Z^\dagger_{(p, W)} 
    \label{eq:Z3-stabilizers}
\end{aligned}
\end{equation} The excitations corresponding to $\omega = e^{2\pi i/3}$ and $\omega^* = e^{4\pi i/3}$ eigenstates of $A^{\mathbb{Z}_3}_v$ are labeled $e_1$ and $e_2$, respectively. Similarly, the $B^{\mathbb{Z}_3}_p$ violations are $m_1$ and $m_2$. The $\mathbb{Z}_3$ toric code has a charge-conjugation symmetry, which swaps the qutrit states $\ket{1}$ and $\ket{2}$. At the anyon level, this is reflected in the fact that re-labeling $e_1 \to e_2$ and $m_1 \to m_2$ preserves all the physical properties of the code. We can introduce qubit gauge degrees of freedom to realize a local version of the charge conjugation symmetry, yielding the $\alpha_v$ term of $S_3$. Note that the product of all $\alpha_v$ terms gives the global charge conjugation operator. The way that the existing stabilizers of the $\mathbb{Z}_3$ toric code have to be modified to commute with $\alpha_v$ yields $S_v, S_p$. 

We can return to the $\mathbb{Z}_3$ toric code from $S_3$ by measuring the qubits of $S_3$ in the $\z$ basis. This can be done globally by measuring all qubits, or locally by only measuring qubits in a small patch. Measuring $\z$ projects the qubits onto a particular configuration of $\beta_p$ stabilizers, and glues together the $\alpha_v$ stabilizers across the entire patch being ungauged. Additionally, the $S_v$ and $S_p$ operators are projected onto modified versions of the stabilizers in Eq. \ref{eq:Z3-stabilizers}, with the daggers in different places depending on the qubit measurement outcomes.

\subsubsection{Ungauging Computational Anyons}

We are interested in running the decoding in parallel with a computation; we may ask how the ungauging and re-gauging operations impact any computational anyons living in a committed region. Do the measurements we are performing read out any logical information? No: crucially, logical information is encoded non-locally between well-separated computational anyons. Local operations in the neighborhood of a single computational anyon do not have access to the logical state. As long as our decoder never ungauges a region containing multiple computational anyons, it will never reveal any logical information. 

This argument can be made precise in the following way: consider a set of four $\mu$ anyons, divided into two pairs. Suppose we encode a qubit such that $\ket{0}$ corresponds to the pairs fusing to vacuum, and $\ket{1}$ corresponds to the pairs fusing to $e$ (see Fig. \ref{fig:gauging}a). Logical operations must encircle at least two of the computational anyons: $Z$-like logical operators encircle the pairs with definite fusion, while $X$-like logicals encircle one anyon from each pair. These logicals can be deformed to enclose an abitrary region, so long as the number of computational anyons enclosed remains the same. So long as our ungauging region remains small compared to the distance to any other computational anyon, we can always clean this logical projector away from the support of the ungauging operation. Thus, the logical state must remain unchanged (see Fig. \ref{fig:gauging}b).

\begin{figure}
    \centering
    \includegraphics[width=1\linewidth]{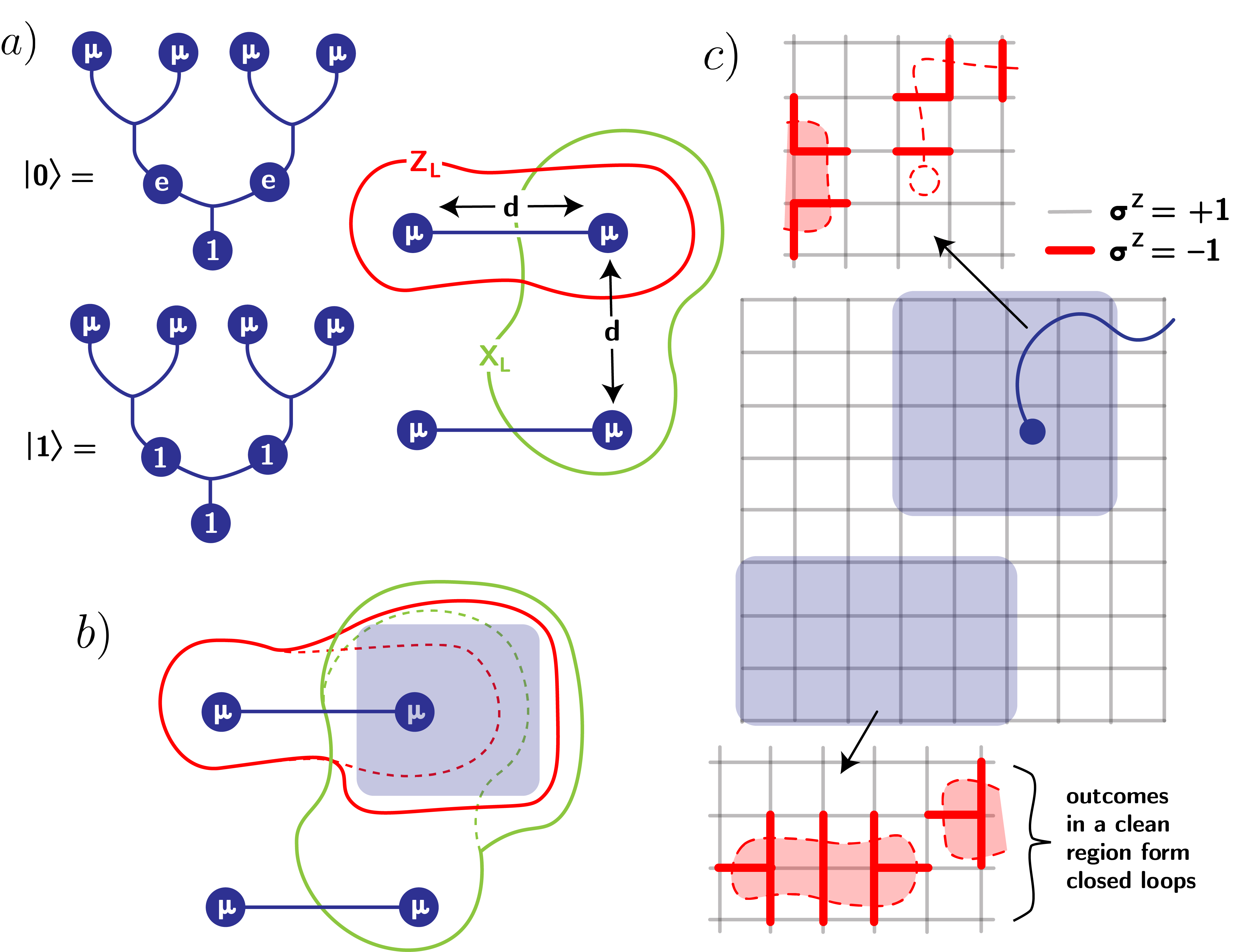}
    \caption{
    a) A logical qubit can be encoded in the fusion outcomes of four $\mu$ anyons: the two states are distinguished by the whether pairs fuse to identity or to $e$. The distance of the logical qubit is determined by the separation between the consistuent $\mu$ anyons, since logical operators must span at least two anyons forming the qubit. b) Logical operators can always be cleaned away from an ungauging region (shaded) containing only one anyon in the logical qubit, meaning the logical information is not disturbed by suitably localized ungauging operations. c) Ungauging is performed by measuring $\sigma^Z$ on every edge in a given region (highlighted in purple). The region without any $\mu$ anyons (bottom) will yield $\sigma^Z = -1$ outcomes that form closed loops, reflecting the fact that all $\beta_p$ stabilizers were satisfied prior to the measurement. A region with a $\mu$ anyon will have open strings of $-1$ outcomes, terminating at the locations of the $\mu$ anyons. }
    \label{fig:gauging}
\end{figure}

\subsubsection{Boundary Detectors}

\begin{figure}
    \centering
    \includegraphics[width=0.7\linewidth]{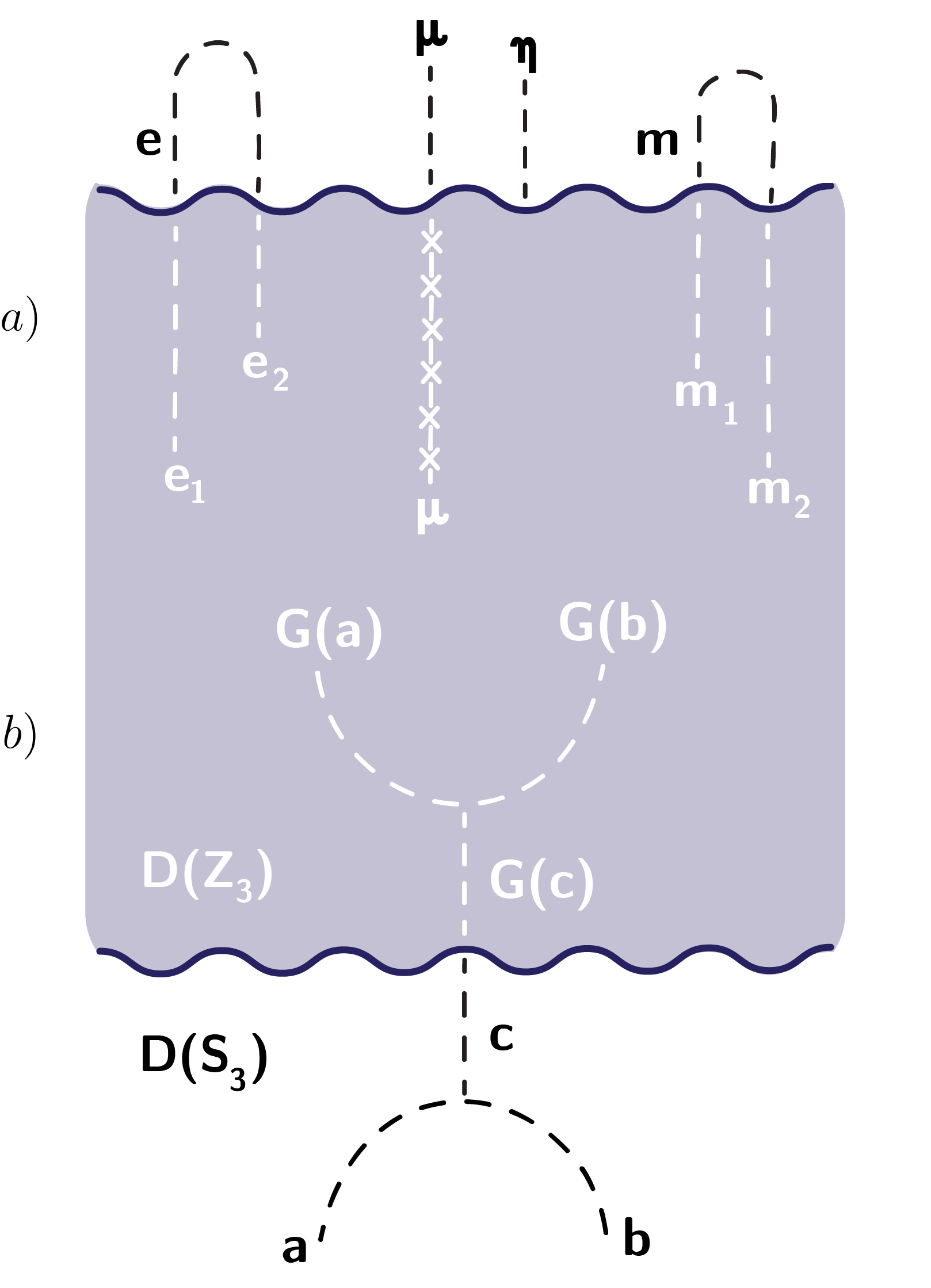}
    \caption{a)~ The domain wall between $D(S_3)$ and $D(\mathbb{Z}_3)$. Crossing into $D(\mathbb{Z}_3)$, $e$ and $m$ anyons are split into two kinds of Abelian anyons, while $\mu$ anyons become confined defects (indicated by the crosses), and $\eta$ anyons become indistinguishable from the vacuum.  b)~ Anyon fusion is consistent across the domain wall: if two anyons $a$ and $b$ fuse to $c$ in $D(S_3)$, their counterparts $G(a)$ and $G(b)$ will fuse to $G(c)$ in $D(\mathbb{Z}_3)$.}
    \label{fig:domain-wall}
\end{figure}

Because the $D(\mathbb{Z}_3)$ and $D(S_3)$ stabilizers are simply related by single site $\z$ measurement, there is a clear correspondence between violations of stabilizers on either side of the gauging map, and thus between anyons in the two codes--- see Fig. \ref{fig:domain-wall}a. This makes the construction of detectors spanning the boundary between $D(S_3)$ and $D(\mathbb{Z}_3)$ straightforward to construct: we merely compare the measurement outcomes for stabilizers and their counterparts on either side. In particular, suppose we have performed the ungauging measurements at time $t = t^*$. Then we can construct the boundary detectors for $\mu, e, m$ anyons as follows:
\begin{equation}
    D^{\mu}_p(t^*) = \beta_p (t^*) - \beta_p(t^* - 1)
\end{equation}
\begin{equation}
\begin{aligned}
    D^{e/m}_{v/p} (t^*) = S^{\mathbb{Z}_3}_{v/p}(t^*) - S_{v/p}(t^* - 1)
\end{aligned}
\end{equation}
Note that $\beta_p(t^*)$ can be inferred from the single site $\sigma^Z$ measurements that implement the ungauging operation. 

Errors that occur before the ungauging timestep will not light up these detectors, since the eigenvalues of the $D(S_3)$ and $D(\mathbb{Z}_3)$ stabilizers on either side of the ungauging boundary will be the same. The detectors will only light up should a measurement error cause the two stabilizers to disagree, or if a new physical error occurs at the same time as the ungauging.

The $\eta$ anyons in the vicinity of gauging boundaries are a little different--- they do not have detectors spanning the boundary as the $\alpha_v$ stabilizers have no counterparts in the $D(\mathbb{Z}_3)$ phase. The single site $\sigma^Z$ measurement removes the $\alpha_v$ in the ungauged region from the stabilizer group. This means that $\eta$ anyons can be absorbed and emitted by the gauging boundary between $D(S_3)$ and $D(\mathbb{Z}_3)$, and the number of $\eta$ anyons will fluctuate in the vicinity of a gauging boundary. This is not catastrophic, though; $\eta$ cannot absorb other anyons, and so we can wait to decode them globally until the end of computation.

\subsubsection{Error Propagation and Correction}

Now that we have discussed both the ungauging procedure and the detectors that live at the boundary between $D(S_3)$ and $D(\mathbb{Z}_3)$, we can describe how errors are transformed under gauging operations and how we can error correct the ungauging measurements.

We will modify microscopic ungauging process depending on whether a given region contains $\mu$ anyons or not, as they will impact the ungauging measurement outcomes. First, consider the case where all the $\beta_p$ stabilizers were satisfied in a given region prior to ungauging. After ungauging, the $\z = -1$ measurement outcomes should form closed loops (see Fig. \ref{fig:gauging}c) --- this can be used to error correct the measurement outcomes to deal with any faulty readout. Once the measurement outcomes have been error corrected, charge conjugation operators can be applied to the qutrits according to the final loop configuration, correcting the measured qutrit stabilizers to the form of Eq. \ref{eq:Z3-stabilizers}.

If there are $\mu$ anyons in the region prior to ungauging, the $\z$ measurement outcomes will not all form closed loops. Instead, there will be open strings that can terminate at plaquettes where the $\mu$ anyons are located. The modified qutrit stabilizers along these open strings cannot be corrected to their clean form by a finite-depth circuit (reflecting the fact the $\mu$ anyon is still a non-Abelian excitation, albeit a confined defect, in $\mathbb{Z}_3$). This is not an issue: we can just start measuring the clean stabilizers instead, as any $e$ and $m$ anyons created will be localized to the ungauged region and will be cleaned up along with the pre-existing excitations.

It is crucial to note that the ungauging process preserves fusion outcomes--- see Fig.~\ref{fig:domain-wall}b. This means that a neutral cluster of excitations will remain neutral after they have been ungauged. A cluster with nontrivial anyon content will still have the corresponding anyon content once ungauged. This fact underpins the utility of ungauging for error correction: once excitations are ungauged, their total fusion product can be determined from \emph{local} stabilizer measurements. The process of ungauging reveals previously inaccessible information by uncovering absorbed anyons.

The first round of $D(\mathbb{Z}_3)$ stabilizer measurements allows us to evaluate the overall neutrality of the ungauged region. More specifically, by measuring the product of all $A^{\mathbb{Z}_3}_v$ ($B^{\mathbb{Z}_3}_p$) stabilizers in the ungauged region, we can determine the total $e$ ($m$) anyon content. However, given the possibility of measurement errors, any excitations in the vicinity of this round of stabilizer measurements must be allowed the possibility of pairing into it. As such, we should measure the age of any excitation in the ungauged region relative to the ungauging boundary. Once the decoder is ready to re-gauge, the qubits that have been measured out can be reset to $\ket{0}$, and re-entangled with the qutrits by measuring the $D(S_3)$ stabilizers again.

Going forward, it will be useful to define the notion of an `absorber', which is simply any object which can absorb anyons. Spatial and temporal boundaries of the system, ungauging and gauging domain walls, and non-Abelian anyon wordlines are all absorbers. Given that errors can create non-Abelian anyons, we assume groups of errors are absorbing objects as well. The `absorbing region' associated with an error cluster encompasses the entire spacetime region containing the worldlines of the associated anyons. When we ungauge the region containing some group of errors, we terminate its ability to absorb other anyons.

\subsection{Threshold theorem}

In this section we prove that the probability of a logical error of a computation can be suppressed exponentially quickly in the size of the system, assuming the error rate of circuit elements is sufficiently weak. We first define the noise model we consider, then cover the function of the just-in-time decoder more precisely, including how ungauging is used to make corrections, before proving our result.

\subsubsection{Error model and spacetime unit cells}

The computation is carried out on a circuit that can be expressed in a three-dimensional spacetime. This spacetime is divided into small unit cubes with detectors and circuit elements associated to these spacetime units. This division is such that every circuit element and detector is associated to exactly one unit cube. Unit cubes are local such that if an error occurs on a given unit cube, it will only trigger detection events on nearby unit cubes within a small constant radius of where the error has occurred. We model errors on individual circuit elements although in practice we coarse grain the problem such that we are only interested in whether or not any errors occurred on any of the circuit elements associated to a given unit cube. Provided we assume uniform unit cubes with some constant number of circuit elements associated to them, $N$, each of which fail and introduce an error with (let's say for simplicity) uniform probability $\varepsilon$, then we have the probability that a unit cube experiences an error $p$ such that 
\begin{equation}
p = 1 - (1-\varepsilon)^N. 
\end{equation}

The action of the local error channel on our system is given by: 
\begin{equation}
    \mathcal{E}[\rho] = \sum_{E} p^{|E|} \hat{E} \rho \hat{E}^\dagger
\label{eq:error-model}
\end{equation} Here, $E$ is an error configuration; as a mathematical object, it is a set whose elements are locations on the lattice where an error occurred. $\hat{E}$ is the actual operator. $|E|$ is the weight of the error configuration (the number of elements of the set $E$), and $p$ is the error rate. For simplicity, we will call $E$ an error.

An interesting quirk of just-in-time decoding is that the error-correction circuit is adaptive, in the sense that the specific circuit that is conducted at a given spacetime unit cube may depend on errors that have occurred in the past, and it is not {\it a priori} obvious that all choices of circuit will experience the same noise. This may lead to time-like correlations in the noise model. An easy way to deal with this issue is to assume that all choices of circuit performed at a given spacetime unit all experience the same noise. If this assumption is not satisfying to the reader, one can modify the noise model as follows. We simply assume that all unit cubes experience the noise of all possible circuits they could be running in tandem. This upper bounds the noise that each unit cube will experience, independent of the circuit it runs. Moreover, since there are only a small finite number of constant depth circuits a given unit cube may be running, there is still only a finite probability $p$ of an error occurring at a given unit cube, even if we make this conservative upper bound to the noise a unit cube experiences. As such we can still find a lower bound to the threshold error rate.

\subsubsection{Error clustering}

Before describing the details of the just-in-time decoder, we first characterize some important properties of the local noise model we are considering. In particular, we will heavily use the fact that, when the noise is sufficiently small, error configurations consist of small, well-separated `chunks' \cite{Bravyi_2013}. Anyons belonging to such a chunk must be neutral, having been created by a local operator. The just-in-time decoder identifies these neutral clusters of excitations (by ungauging) and corrects them within a small enough neighborhood that they never talk to each other. For an illustration of this idea, see Fig. \ref{fig:app-spacetime}a--- which depicts an example spacetime error configuration--- and Fig. \ref{fig:app-spacetime}b,  which demonstrates how the just-in-time decoder deals with such a configuration.

Consider a spacetime lattice where errors may occur on edges living in the $x-y$ plane (physical errors) or edges in the $z$ direction (measurement errors). Such errors light up detectors, which live on sites. We defined our error model in Eq. \ref{eq:error-model} in terms of error configurations $E$: recall that an error configuration is a set whose elements are the edges where errors occurred. Going forward, we the $L_\infty$ metric to measure the distance between two locations $u, v$ in spacetime: $|u-v|$ will be given by the maximum distance between the two points in any direction.

The following terminology will be useful: the diameter of a subset $M \subseteq E$ is the maximum distance between any two elements of $M$. The subset $M$ is $r$-connected if it cannot be decomposed into disjoint subsets separated by more than distance $r$. Conversely, two disjoint subsets $M, N \subseteq E$ are $r$-separated if the minimum distance between an element of $M$ and an element of $N$ is $r$. The $r$-neighborhood of $M$ is the set of sites and edges up to distance $r$ away from $M$, including the edges in $M$ itself.

\begin{figure*}
    \centering
    \includegraphics[width=0.75\linewidth]{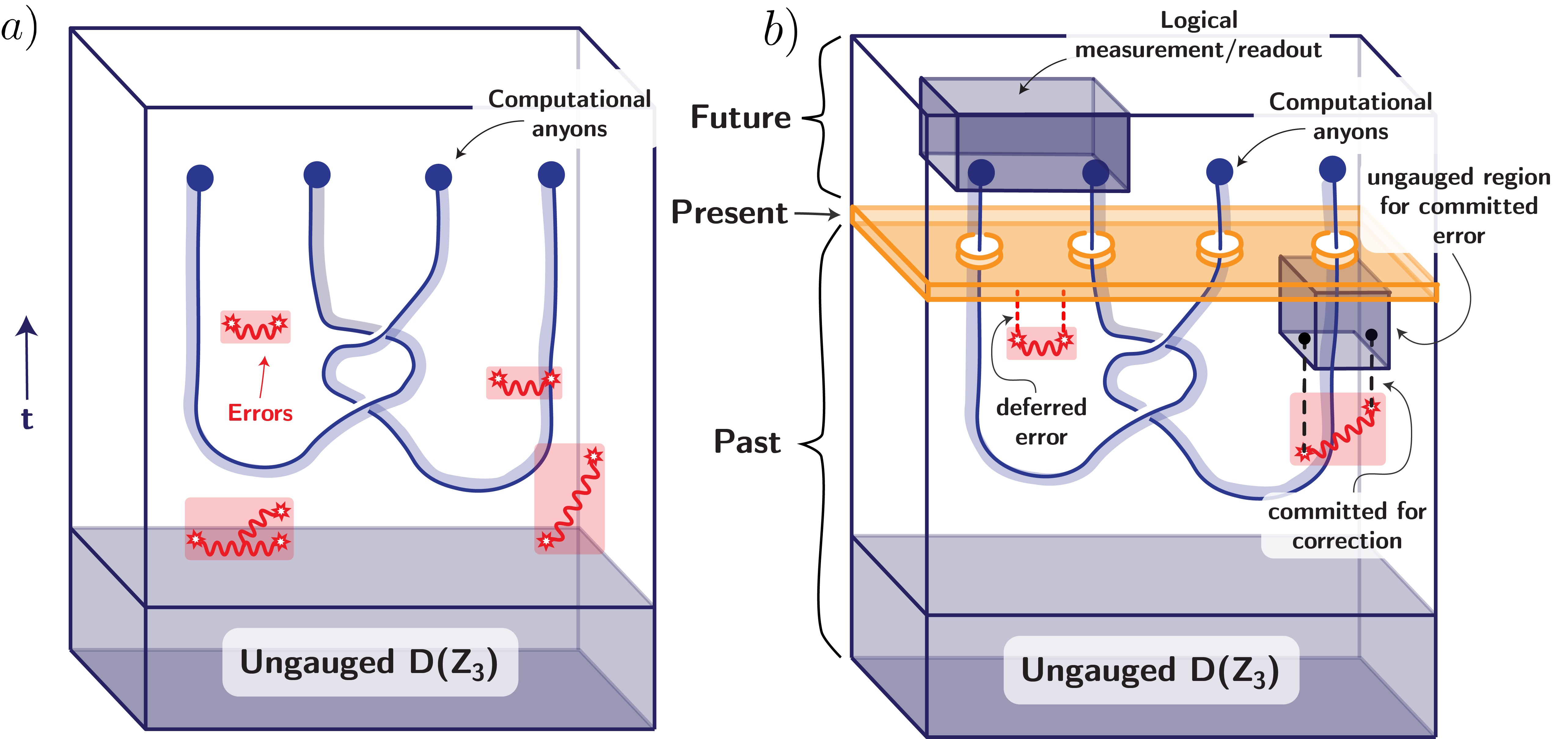}
    \caption{a)~ A spacetime configuration includes computational anyons (blue) and errors (red). Computational anyons are kept far apart to maintain code distance. In the detector model, errors are string-like, with detection events (anyons) at their endpoints. b)~ The just-in-time decoder takes into account the past spacetime history to make decisions about the present correction step. Errors that have lived long enough can be committed for correction, which involves first ungauging to terminate the absorbing region and reveal the fusion outcome of the error cluster. Errors that are too young are deferred until a later timestep. Computational anyons are protected from readout as long as the ungauged regions contain at most a single computational anyon.}
    \label{fig:app-spacetime}
\end{figure*}

\begin{figure}
    \centering
    \includegraphics[width=0.75\linewidth, origin=c]{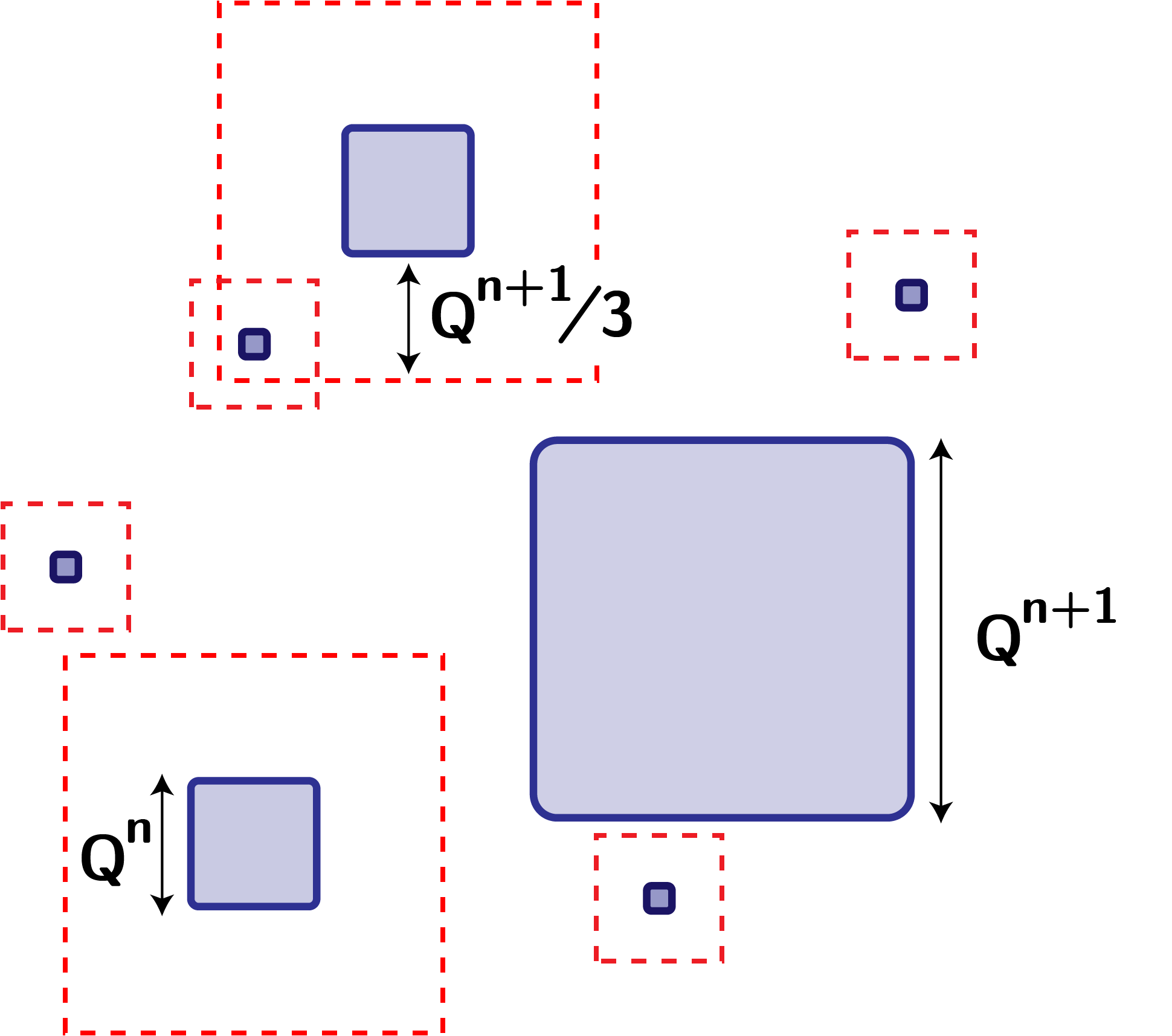}
    \caption{The chunk decomposition. All errors in a given error configuration can be contained in the blue boxes of size $Q^n$ for $n = 0, \dots, \sim \log L$, where $L$ is the linear size of spacetime. Each box of size $Q^n$ has a buffer region of size $(1/3)Q^{n+1}$ surrounding it (red dashed lines). No box of size equal or greater may overlap the buffer region, although smaller boxes may be closer.}
    \label{fig:chunk-decomp}
\end{figure}

Now we introduce the `chunk' decomposition, following the construction of Bravyi and Haah \cite{Bravyi_2013}. The idea is to split an error $E$ into a hierarchy of subsets of larger and larger sizes. Later on, our decoder will be designed to address the different levels of the hierarchy separately, preventing errors from spreading past a given length scale. 

\begin{definition}[Chunk Decomposition]
Given an error $E$, and an integer $Q \geq 2$, we define \textbf{chunks} of this error recursively: 
    \begin{itemize}
    \item A \textbf{level-0} chunk is a single site with support in $E$. Denote the set of all such sites as $E_0 = E$.
    \item A \textbf{level-n} chunk is the disjoint union of two level-(n-1) chunks with diameter $\leq Q^n/2$. The set of sites belonging to a level-$n$ chunk is given by $E_n \subseteq E$.
\end{itemize} An error $E$ thus has a collection of subsets: \begin{equation}
    E = E_0 \supseteq E_1 \supseteq \cdots \supseteq E_m
\end{equation} with $E_{m+1} = \varnothing$ for some integer $m$. We define the sets $$F_n = E_n \setminus E_{n+1},$$ which are the subsets of $E$ that take part in a level-$n$ chunk but no higher chunk. Then the \textbf{chunk decomposition} of $E$ is given by: \begin{equation}
    E = F_0 \cup F_1 \cup \cdots \cup F_m
\end{equation} $E$ is called a level-$m$ error if the decomposition terminates at $F_m$.
\end{definition}

We will be interested not only in the sets $F_n$, but in \emph{connected subsets} $F_{n, \alpha} \subseteq F_n$. We will call a $Q^n$-connected component $F_{n, \alpha} \subseteq F_n$ a \textbf{level-$n$ nugget}. In particular, the following separation lemma can be proven for nuggets \cite{Bravyi_2013} (see Fig. \ref{fig:chunk-decomp} for a visualization):
\begin{lemma}[Nugget Separation Lemma]
    For $Q \geq 6$, a level-$n$ nugget must be separated from nuggets of the same size or larger by a distance of at least $(1/3) Q^{n+1}$. Stated differently:
    \begin{equation}
        |F_{n, \alpha} - F_{m, \beta}| \geq (1/3)Q^{\max(n, m)}
    \end{equation} where $F_{n, \alpha}$ and $F_{m, \beta}$ are distinct nuggets. 
    \label{lem:nugget-sep}
\end{lemma}

So far, we have focused on a physical error $E$. However, the input to any decoder won't be the actual error $E$, but rather its syndrome $\sigma(E)$. The syndrome $\sigma(E)$ is the set of all detectors violated by a physical error $E$. Since we are considering a local topological code, $\sigma(E)$ must lie in the $1$-neighborhood of $E$. We refer to the set of defects created by a level-$n$ nugget $F_{n, \alpha}$ as a \textbf{level-$n$ cluster} $B_{n, \alpha}$.

It is important that the cluster sizes terminate at some point, ideally such that the maximum cluster size is much smaller than the size of the system. This will ensure that no single nugget is large enough to lead to a logical error, as long as the correction applied by the decoder is sufficiently localized. In \onlinecite{Bravyi_2013}, it was proven that, given the error rate goes as $p \sim 1/(3Q)^6$, the probability of a cluster of size $Q^n$ decays doubly-exponentially with $n$. Our proof strategy will be to find $Q$ such that our chunk decomposition has the properties we need, which will then yield a lower bound on the threshold error rate.

\subsubsection{Just-in-time decoding}

We now describe the just-in-time decoding algorithm itself. In the simplest terms, the just-in-time decoder is a set of rules that tell us how to fuse groups of erroneous anyons. It must do this on the fly, as these anyons may be hiding the presence of other excitations that have been absorbed. For $D(S_3)$ these are the excitations labeled $\mu$, $e$ and $m$. We therefore consider fusing groups of $\mu$, $e$, and $m$ excitations to reveal excitations these anyons may be hiding. 

Unlike the other anyon types, $\eta$ anyons cannot hide any other types of excitations. We can therefore deal with all the $\eta$ detection events globally once the computation is complete. So, going forward, we will assume the just-in-time decoder is focusing solely on the $\mu, e$ and $m$ anyons.

The just-in-time decoder needs to reveal absorbed anyons as quickly as possible while maintaining fault-tolerance. We can strike this balance by enforcing a commit rule that only allows corrections once the decoder is reasonably sure that a given set of excitations exist. At each timestep $T$, the just-in-time decoder looks at the syndrome history for $t \leq T$, and sorts excitations into two categories: (i) excitations that are younger than the distance to the nearest absorber, and (ii) excitations that are as old or older than that distance. The first category is deferred until a later timestep, when the decoder can be more sure of how to correct it, while the latter category can be corrected using ungauging. 

A high-level description of the decoder is as follows:
\begin{enumerate}
    \item At timestep $T$, sort clusters of excitations into two categories: `defer' and `attempt to correct'. Clusters with any constituent excitations younger than the size of the cluster are deferred. Clusters where all members are at least as old as the size of the cluster are put into `attempt to correct'. 
    \item Reverse the most recent measurement outcome to move the excitation to the next time step if it is in the `defer' category.
    \item For clusters in the `attempt to correct' category, we have three scenarios: (i) the cluster lives in $D(S_3)$, (ii) the cluster lives in an ungauged $D(\mathbb{Z}_3)$ region, and (iii) it spans a boundary between the two.
    \begin{enumerate}
        \item If (i) or (iii), ungauge such that the full cluster lives in $D(\mathbb{Z}_3)$. 
        \item If (ii), determine whether the cluster is neutral or not. If the ungauged cluster is neutral, the decoder corrects it. If the cluster is not neutral, it must be part of a larger cluster the decoder has not identified yet, and should be deferred until a later time.
    \end{enumerate}
    \item Repeat for the next timestep.
\end{enumerate}
We can prove that such a decoder must have a threshold for $D(S_3)$.

\subsubsection{Threshold Proof}


\begin{figure}
    \centering
    \includegraphics[width=\linewidth]{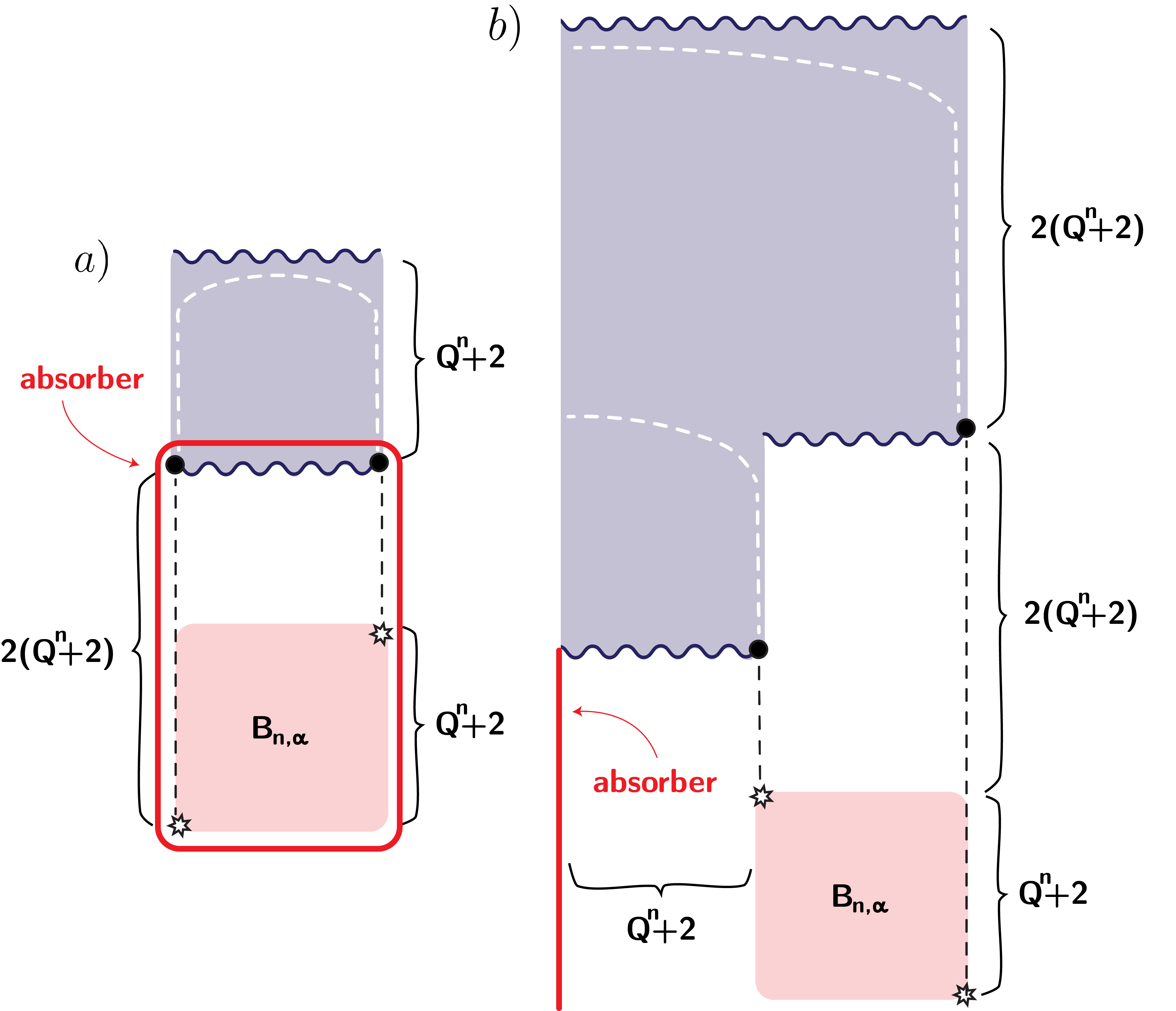}
    \caption{a)~ Far from any other absorbers, a single isolated group of defects of size $Q^n + 2$ is corrected by the decoder in time at most $3(Q^n + 2)$. If the defects lie at the extremal points pictured, the defect in the bottom left will have to wait $2(Q^n + 2)$ timesteps before the other defect has lived long enough for the decoder to commit to a correction. At this point, the ungauging is performed, terminating the absorbing region associated with $B_{n, \alpha}$. The decoder spends $Q^n + 2$ more timesteps in the $D(\mathbb{Z}_3)$ phase before correcting, and re-gauging. b)~ If a cluster is close enough to another absorber, it will be corrected in at most $5(Q^n + 2)$ timesteps. A delay of $Q^n+2$ is possible when the older defect is further from the absorber, since the size of the ungauged region will be $2(Q^n + 2)$ rather than $Q^n + 2$. }
    \label{fig:extremal-isolated-cluster}
\end{figure}

Before we proceed with the first part of the proof, we begin with by proving a simpler lemma that will help us prove the general case. We first show that a single error cluster that is well isolated from all other clusters of errors is corrected within a bounded region of spacetime.

\begin{lemma}
    An isolated nugget $F_{n, \alpha}$ is corrected within the  $4Q^n+7$ neighborhood of $B_{n, \alpha}$, and no defect in $B_{n, \alpha}$ will live longer than $5(Q^n + 2)$.
    \label{lem:isolated-chunk}
\end{lemma}
\begin{proof} We consider two cases: (1) $F_{n, \alpha}$ is at least $3(Q^n+2)$-separated from an absorber, and (2) $F_{n, \alpha}$ is $3(Q^n+2)$-connected to another absorber.
    
    (1) The defects belonging to $B_{n, \alpha}$ are at most $Q^n + 2$ apart, and must be collectively neutral since they were created by an isolated, local error $F_{n, \alpha}$. If $B_{n, \alpha}$ contains defects at extremal points (see Fig. \ref{fig:extremal-isolated-cluster}a), we will have to wait at most $Q^n + 2$ time steps after the youngest defect before we can commit to ungauging the cluster. Furthermore, we will have to wait another $Q^n+2$ time steps before we can pair the defects inside the ungauged region, as their age must now be measured relative to the ungauging boundary. So the correction of $F_{n, \alpha}$ lies within the $2Q^n + 3$ neighborhood of $B_{n, \alpha}$. The oldest defect in the cluster will have to wait for a time of $3(Q^n + 2)$ before being corrected. Since we are further than $3(Q^n+2)$ from any other absorber, it will never be permissible to merge the cluster into the absorber.

    (2) The defects created by $F_{n, \alpha}$ are neutral and lie in the $1$-neighborhood of $F_{n, \alpha}$. Thus, defects created by  $F_{n, \alpha}$ are at most $Q^n + 2$ apart and at most $4(Q^n + 2)$ separated from the other absorber. In the worst-case scenario, the youngest defect in the cluster is the furthest from the other absorber, and the correction will lie in the $4Q^n+7$ neighborhood of $B_{n, \alpha}$. This correction will take time $5(Q^n + 2)$ from the appearance of the first defect in $B_{n, \alpha}$. See Fig. \ref{fig:extremal-isolated-cluster}b for an illustration.

    Thus, the upper bound on the size of the error spread and the time it takes the decoder to correct all defects is given by $4(Q^n + 2)$.
\end{proof}

We have shown that for an \emph{isolated} level-$n$ nugget, all associated defects are corrected in time $t \leq 4(Q^n + 2)$, either by merging them to each other or with a nearby absorber. If we now have an error $E$ containing many nuggets of different sizes, the decoder may now pair defects originating from different nuggets, when smaller clusters occur close enough to larger ones. Suppose we have a level-$n$ cluster $B_{n, \alpha}$ we have decided to commit to ungauging. Smaller clusters up to level-$(n-1)$ might merge with $B_{n, \alpha}$ if it is permissible to do so before it becomes allowed to pair within the smaller cluster, since the defects in $B_{n, \alpha}$ are suddenly around the same age as the smaller cluster defects (see Fig. \ref{fig:linking}a). Our goal is to prove the decoder does not merge too many clusters into a logical error.

To start, we define the notion of \emph{linking}, which quantifies when smaller clusters may need to be merged with larger ones by the decoder. This happens when a small cluster is too close to the absorbing region created by the larger cluster (see Fig. \ref{fig:linking}a, b). 

\begin{figure}
    \centering
    \includegraphics[width=1\linewidth]{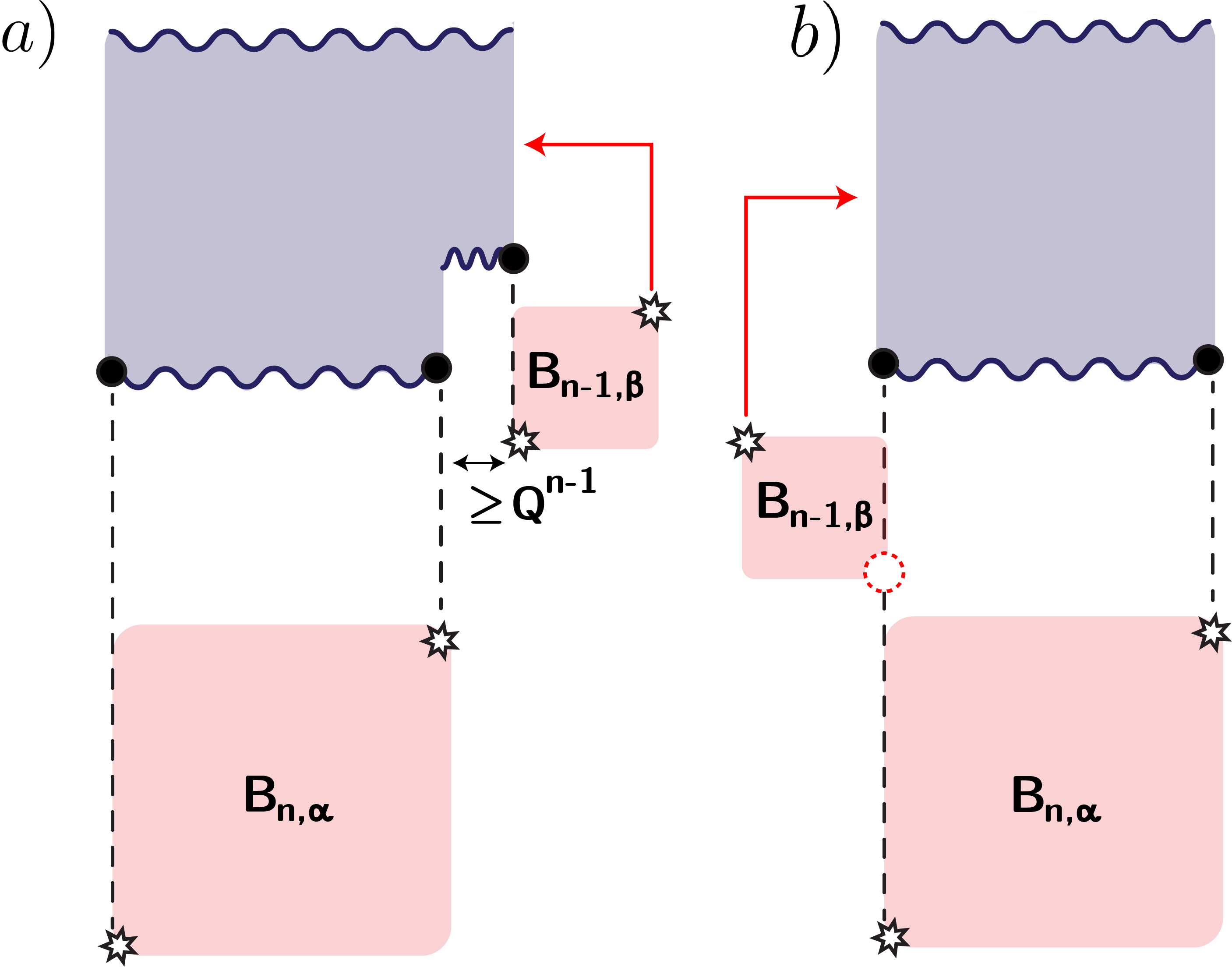}
    \caption{We illustrate two scenarios that cause a smaller cluster to link with a larger one. a) If a small cluster $B_{n-1, \beta}$ is closer than $Q^{n-1}$ to the ungauging boundary for a larger cluster $B_{n, \alpha}$, the decoder may pair an excitation in $B_{\beta}$ with one in $B_{\alpha}$. This is possible because any excitations in the ungauged region have age measured relative to the ungauging timestep. The remaining $B_\beta$ excitation must eventually be merged into $B_\alpha$ to complete the correction. b) As another linking example, consider a defect in the smaller cluster $B_{n-1, \beta}$ that coincides in space with a defect from $B_{n, \alpha}$. The defect in $B_{n-1,\beta}$ may be absorbed without lighting up a detector. This means the cluster $B_{n-1, \beta}$ is non-neutral, according to detector readings. The solution is to recognize that $B_{\alpha}$ and $B_{\beta}$ must be corrected together. }
    \label{fig:linking}
\end{figure}

\begin{lemma}[Linking]
    Suppose a level-$n$ cluster $B_{n, \alpha}$ associated with a nugget $F_{n, \alpha}$ lies in the $2(Q^{n} +2)$ neighborhood of the absorbing region associated with another cluster $B_{n+j, \beta}$, with $j \geq 1$. Once this absorbing region is terminated by ungauging, all defects in $B_{n, \alpha}$ will be included in the ungauged region within time $t \leq 3(Q^n + 2)$. We say $B_{n, \alpha}$ is \emph{\textbf{directly linked}} to $B_{n+j, \beta}$.
    
    If $B_{n, \alpha}$ lies outside of the $2(Q^{n} +2)$ neighborhood of the absorbing region, it will not be merged with the larger ungauged region. 
    
\end{lemma}
\begin{proof}
    Suppose the boundary of the absorbing region is temporally and spatially separated from $B_{n, \alpha}$ by $2(Q^{n} +2)$. Then some other defects in $B_{n, \alpha}$ may be maximally separated from the boundary by $3(Q^{n} +2)$ (temporally and spatially). It will take time $2(Q^n + 2)$ to commit to ungauging the region connecting the inner defects to the boundary, and then a further time $Q^n + 2$ to commit to ungauging the region containing the outer defects. 

    On the other hand, if the boundary is further than $2(Q^{n} +2)$ away from $B_{n, \alpha}$, the decoder will commit to ungauging just $B_{n, \alpha}$ faster, within time $2(Q^{n} +2)$ from the first appearance of a defect in this cluster. 
\end{proof}

We want linking to be relatively benign, meaning that there are bounds on how many other clusters may be linked with a given cluster. Ideally, we want to be able to regroup nuggets and their associated clusters and still have a valid chunk decomposition at the end (with a different effective $Q$ perhaps). The following set of results put bounds on the size of $Q$ that lead to nice linking properties.

\begin{theorem}[Linking properties]
    Given $Q \geq 60$, the following all hold: 
    \begin{enumerate}
        \item A cluster $B_{n, \alpha}$ can not be directly linked to another cluster of the same size.
        \item A given cluster $B_{s, \alpha}$ is directly linked to only one larger cluster $B_{n, \beta}$ ($n > s$). 
        \item An absorbing region associated with a cluster of size $Q^n$ can only have one directly linked cluster $B_{p, \gamma}$ of each size $p \leq n$.
    \end{enumerate}
    \label{th:linking-properties}
\end{theorem}
\begin{proof}
    (1) The $2(Q^{n}+2)$ region around an absorbing region of size $Q^n$ is within the $3(Q^{n}+2)$ neighborhood of some nugget $F_{n, \beta}$. We want to make sure that this neighborhood cannot have overlap with a defect belonging to another nugget $F_{n, \alpha}$ of the same size. Since defects in $B_{n, \beta}, B_{n, \alpha}$ must be separated by $Q^{n+1}/3 - 2$, this requires that \begin{equation}
        Q^{n+1}/3 - 2 \geq 3(Q^{n}+2) ~\forall~ n \geq 0 \implies Q\geq 33
    \end{equation}
    
    (2) Suppose that we have two clusters $B_{n, \beta}$, $B_{m, \gamma}$, with $m \geq n > s$, and suppose that $B_{s, \alpha}$  is directly linked to both. Then, using the triangle inequality, we must have the following: 
    \begin{equation}
    \begin{aligned}
         d(F_{s,\alpha}, F_{n, \beta}) &\leq 3(Q^s + 2) \\
         d(F_{s,\alpha}, F_{m, \gamma}) &\leq 3(Q^s + 2) \\
         \frac{Q^{n+1}}{3} - 2\leq d(F_{n, \beta}, &F_{m, \gamma}) \leq 6(Q^s + 2) \implies Q \leq 60
    \end{aligned}
    \end{equation} For $Q \geq 60$, this will never be possible. 
    
    (3) Any linked clusters $B_{p, \gamma}$, $B_{p, \delta}$ must lie in the $2(Q^p+2)$ neighborhood of the ungauging boundary, which means they can be separated by at most $4(Q^p + 2)$. But by the chunk decomposition, they must be at least distance $Q^{p+1}/3 - 2$ apart. We require that for all $p \geq 0$: \begin{equation}
        Q^{p+1}/3 -2 > 4(Q^{p} + 2) ~\forall~ p < n \implies Q \geq 43
    \end{equation} 
Thus, the smallest $Q$ necessary for all these properties to hold is $Q = 60$.
\end{proof}

The above theorem allows us to define the notion of a \emph{linked tree}, which captures the notion that smaller clusters are linked to larger clusters that are themselves linked to even larger clusters, etc. We want to bound the depth of a linked tree that terminates at a given cluster size, as this will allow us to bound the length of time it takes to correct all clusters that take part in the linked tree. We can then apply techniques for proving thresholds for constant spread errors, having proven the constant spread for the linked trees (rather than the base clusters).

\begin{definition}[Linked Tree]
    The set of clusters $\mathbf{B} = \{B_{k_i, \alpha_i}\}$, $k_i \in \{ 0, \cdots, m\}$ is a \emph{\textbf{linked tree}} if it can be formed by starting with some cluster $B_{k_{\max}, \alpha}$ of size $Q^{k_{max}}$, with $k_{max} = \max\{k_i\}$ and adding all its directly linked clusters, and then all directly linked clusters to those directly linked clusters, etc. We denote such a linked tree $\mathbf{B}_{k_{\max}}$.
\end{definition}

\begin{lemma}
    A linked tree $\mathbf{B}_{k_{\max}}$ has a unique cluster $B_{k_{\max}, \alpha}$. All other constituent clusters are smaller.
\end{lemma}
\begin{proof}
    Suppose there was a linked tree with two largest clusters of the same size. By Lemma \ref{th:linking-properties}(1) these two clusters cannot directly linked. By Lemma \ref{th:linking-properties}(2) any cluster can only be linked to one cluster of a larger size. However, in order for the two largest clusters to be \emph{indirectly} linked, there must be a path of links along which some smaller cluster is linked to two larger clusters. So there is never a way for two clusters of the same size to be linked, even indirectly via a set of intermediate links. 
\end{proof}

\begin{lemma}
    A linked tree $\mathbf{B}_{k_{\max}}$ has diameter $|\mathbf{B}_{k_{\max}}| \leq 2(Q^{k_{\max}}+2) + 8 \left (\frac{Q^{k_{\max}} - 1}{Q-1}\right) + 16k_{\max}$. For $Q \geq 11$, we have $|\mathbf{B}_{k_{\max}}| \leq 3Q^{k_{\max}}$.  
    \label{lem:tree-size}
\end{lemma}
\begin{proof}
    By Lemma \ref{th:linking-properties}(2), we know that a linked tree can only contain monotonically decreasing linked clusters, since a given cluster can only be directly linked to one larger cluster. Thus, given $k_{\max}$, the maximal `depth' of the linked tree is $k_{\max}$, supposing there is a path of links from the largest cluster to a cluster of size $0$, including all possible intermediate cluster sizes. 
    
    Since each cluster $B_{j, \alpha_j}$ along the path must be within the $3(Q^j +2)$ neighborhood of the next largest cluster $B_{j+1, \alpha_{j+1}}$ (it must be within $2(Q^j+2)$ of the absorbing region belonging to $B_{j+1, \alpha_{j+1}}$), $B_{j, \alpha_j}$ may extend the size of the neighborhood containing the linked tree by $4(Q^j + 2)$. Summing from $j=0$ to $j=k_{max}-1$ yields:
    \begin{equation}
        \sum_{j = 0}^{k_{\max}-1} 4Q^j + 8 = 4\left ( \frac{Q^{k_{\max}}-1}{Q-1}\right ) + 8 k_{\max}
    \end{equation} Since the temporal boundary belonging to the largest cluster $B_{k_{\max}}$ is itself within the $2(Q^{k_{\max}}+2)$ neighborhood of $B_{k_{\max}}$, the diameter of the linked tree is bounded by: \begin{equation}
        |\mathbf{B}_{k_{\max}}| \leq 2(Q^{k_{\max}}+2) + 8\left ( \frac{Q^{k_{\max}}-1}{Q-1}\right ) + 16 k_{\max}
    \end{equation}
    We can simplify our bound a bit, since we have the following: 
    \begin{equation}
    \begin{aligned}
        8 &\left (\frac{Q^k - 1}{Q-1} \right ) + 16k + 4 < Q^k \\
        &\implies 9 Q^{k} + 16k - 4 < Q^{k+1} \\
        &\implies Q \geq 11
    \end{aligned}
    \end{equation} So for $Q \geq 11$, we have the desired bound on $|\mathbf{B}_{k_{\max}}|$.
\end{proof}

\begin{lemma}
    If $Q \geq 60$ then two \emph{disjoint} linked trees $\mathbf{B}_{n}$ and $\mathbf{B}_m$ with $m \geq n$ are separated by at least $\frac{1}{4}Q^{n+1} - 2$.
    \label{lem:tree-sep}
\end{lemma}
\begin{proof}
    We know that the linked trees contain unique largest clusters of sizes $n, m$ respectively. Let these clusters be $B_{n,\alpha}$ and $B_{m, \beta}$. These clusters must be separated by at least $Q^{n+1}/3 - 2$. We also require that they are not linked with each other. To ensure this, since by Lemma \ref{lem:tree-size} their associated linked trees have diameter less than three times the size of the largest clusters, we require that: \begin{equation}
    \begin{aligned}
        d(\mathbf{B}_n, B_{m, \beta}) &\geq d(B_{n, \alpha}, B_{m, \beta}) - 2Q^{n} \\
        &\geq \frac{Q^{n+1}}{3} - 2 - 2Q^n \\
        &\geq \frac{\gamma}{3}Q^{n+1} - 2 \implies (1-\gamma)Q \geq 6
    \end{aligned}
    \end{equation} Alternately, we must have $\gamma \leq 1 - 6/Q$. We will just set $\gamma = 3/4$, since we must definitely have $Q > 24$ for the rest of the proof to hold. We also want to make sure this separation is much larger than the diameter of $\mathbf{B}_n$ itself. Conservatively, we can achieve this by satisfying \begin{equation}
        \frac{1}{4} Q^{n+1} - 2 > 9Q^n \implies Q \geq 44
    \end{equation} 
    
    We also want to demonstrate that any cluster $B_{s, \eta} \in \mathbf{B}_m$ must also be similarly far separated from $\mathbf{B}_n$. The above proof suffices for $s \geq n$. For $s < n$, we leverage the fact that $B_{s, \eta}$ must be linked (directly or not) to a cluster of size $p \geq n$. Otherwise, it would not belong to $\mathbf{B}_m$ with $m \geq n$. 
    
    Let the $B_{p, \xi}$ be the smallest cluster linked to $B_{s, \eta}$ while maintaining $p \geq n$. We can bound the distance from $B_{p, \xi}$ to $B_{s, \eta}$ using the monotonicity of cluster sizes along any linked path. Since $p \geq n$, $B_{s, \eta}$ cannot be further than $3Q^{n-1}$ from $B_{p, \xi}$ since we picked the smallest $p$ possible, and so only other chunks of size less than $n$ may link $B_{s, \eta}$ and $B_{p, \xi}$. Thus we have: \begin{equation}
        \begin{aligned}
            d(\mathbf{B}_n, B_{s, \eta}) &\geq d(\mathbf{B}_n, B_{p, \xi}) - 3Q^{n-1} \\
            &\geq d(B_{n, \alpha}, B_{p, \xi}) - 2Q^n - 3Q^{n-1} \\
            &\geq \frac{1}{3}Q^{n+1} - 2 - 2Q^n - 3Q^{n-1} \\
            &\geq \frac{1}{4}Q^{n+1} - 2 \implies Q \geq 26
        \end{aligned}
    \end{equation} So any chunk in $\mathbf{B}_m$ is at least $(1/4)Q^{n+1} - 2$ separated from $\mathbf{B}^n$, proving the desired separation between $\mathbf{B}_m$ and $\mathbf{B}_n$.
\end{proof}

Additionally, we know that the diameter of a tree is bounded, and so the decoder must correct an isolated tree in bounded time:
\begin{lemma}
    An isolated linked tree $\mathbf{B}_n$ is corrected by the decoder in time $t \leq 4(3Q^n + 2)$.
\end{lemma}
\begin{proof}
    Using Lemma \ref{lem:tree-size}, the diameter of $\mathbf{B}_n$ is bounded by $3Q^n$. We can plug this into the result of Lemma \ref{lem:isolated-chunk} to find the desired result.
\end{proof} 

We are now ready to prove that the decoder works given a large enough $Q$.
\begin{theorem}
    Given $Q \geq 89$, the decoder corrects all defects associated with a linked tree $\mathbf{B}_n$ by merging with other defects in $\mathbf{B}_n$. 
\label{th:just-in-time-decoding}
\end{theorem}
\begin{proof}
    We need to satisfy the following inequality:
    \begin{equation}
        \frac{1}{4}Q^{n+1} - 2 > 4(3Q^n + 2) \implies Q \geq 89
    \end{equation} In this regime, the separation between $\mathbf{B}_n$ and any larger tree $\mathbf{B}_m$ is larger than the time it takes to correct $\mathbf{Q}^n$, so the decoder will never merge two disjoint linked trees.
\end{proof}

We now consider the decoding of $\eta$ anyons. What we must prove is that the density of $\eta$ anyons created during the just-in-time decoding is low enough to allow for global decoding to succeed. The most straightforward way to proceed is to assume we will also use an RG-type decoder for the $\eta$ anyons. If we can show that the $\eta$ syndromes have a clustering property, we can ensure that an RG-type decoder will succeed\cite{Bravyi_2013}.

The extra $\eta$ anyons will only be created in the vicinity of error clusters that have been ungauged and re-gauged during the just-in-time decoding process. In particular, they will be restricted to some neighborhood of the linked trees $\{\mathbf{B}_n\}$. Given that the linked trees obey a cluster decomposition given large enough $Q$ (small enough error rate), we can guarantee an RG-type decoder \cite{Bravyi_2013} will suffice to correct the residual $\eta$ anyons. 

\begin{theorem}[Decoding $\eta$ anyons globally]
    Erroneous $\eta$ anyons and $\eta$ anyons created as byproducts of the just-in-time decoding process can be reliably corrected by a global spacetime RG decoder at the end of computation. 

    In essence, we want clusters of $\eta$ anyons to form a chunk decomposition of their own. We can show that clusters of $\eta$ anyons are supported on disjoint regions  $\mathbf{A}_n$ such that, for $Q > 352$, 
    \begin{equation}
        d(\mathbf{A}_{\alpha, n}, \mathbf{A}_{\beta, n}) \geq \frac{1}{8}Q^{n+1}.
    \end{equation}
\end{theorem}
\begin{proof}
    The spacetime history of $\eta$ anyons will lie within the $4(3Q^n + 2)$ neighborhood of the linked trees $\{\mathbf{B}_n\}$. This holds since there are two processes that create $\eta$ anyons: (i) physical or measurement errors, which will by definition lie within the 1-neighborhood of the linked trees, and (ii) random outcomes from the gauging measurements, which must lie within the $4(3Q^n + 2)$ neighborhood, since by Theorem~\ref{th:just-in-time-decoding} all linked trees are corrected within this time. Denote the $4(3Q^n + 2)$ neighborhood of a linked tree $\mathbf{B}_n$ as $\mathbf{A}_n$.

    From Lemma \ref{lem:tree-sep} we know the minimum separation between linked trees. Thus, the minimum separation between $\mathbf{A}_n$ must be
    \begin{equation}
        \frac{1}{4}Q^{n+1} - 8(3Q^n + 2) - 4 > \frac{1}{8}Q^{n+1} \implies Q > 352
    \end{equation}
    
    So if $Q > 352$, we can ensure that the Abelian anyons satisfy a cluster decomposition. 
\end{proof}

The final step in our proof is to determine a lower bound on the error rate such that large error clusters are very rare. As proven in Ref.~\onlinecite{Bravyi_2013}, the probability of clusters of size $Q^n$ occurring is suppressed doubly-exponentially in $n$ given the error rate goes as $p \approx 1/(3Q)^6$. Thus we can ensure that the probability of clusters of size $Q^n \sim O(L)$ (implying $n$ goes as $O(\log L)$) are exponentially suppressed in system size $L$. Given $Q > 352$, we have that $p_{\mathrm{th}} \geq 7.2 \times 10^{-19}$. Thus, we have demonstrated that the threshold error rate of our error correction scheme is greater than zero.

\end{document}